\documentclass[reqno]{amsart}
\usepackage{amsmath}
\usepackage{amsfonts,amssymb}
\usepackage{mathrsfs}

\usepackage{mathtools}

\usepackage{mathpazo}

\usepackage{url}
\usepackage{bbm}

\makeatletter \@addtoreset{equation}{section} \makeatother

\renewcommand\thefigure{\thesection.\@arabic\c@figure}
\renewcommand\thetable{\thesection.\@arabic\c@table}
\newtheorem{theorem}{Theorem}[section]
\newtheorem{lemma}[theorem]{Lemma}
\newtheorem{proposition}[theorem]{Proposition}

\theoremstyle{remark}

\newtheorem{remark}[theorem]{Remark}

\newcommand{\mc}[1]{{\mathscr #1}}

\newcommand{\bb}[1]{{\mathbb #1}}

\newcommand{\<}{\langle}
\renewcommand{\>}{\rangle}

\newcommand{\CC}{\mathbb{C}}

\newcommand{\RR}{\mathbb{R}}

\newcommand{\ZZ}{\mathbb{Z}}

\let\ve=\varepsilon

\let\ve=\varepsilon

\let\n=\eta

\newcommand{\sfrac}[2]{{\smash{\frac{#1}{#2}}}}

\newcommand{\tnorm}{\vert \kern -1pt\vert\kern -1pt\vert}

\newcommand{\plus}{\!+\!}
\newcommand{\minus}{\!-\!}

\usepackage{zito}

\begin{document}
\title[Fractional Superdiffusion] {$3/4$-Fractional superdiffusion in a system of harmonic oscillators perturbed by a conservative noise}
\author{C\'edric Bernardin}
\address{ C\'edric Bernardin\\ Universit\'e de Nice Sophia-Antipolis\\ Laboratoire J.A. Dieudonn\'e\\ UMR CNRS 7351\\ Parc Valrose\\ 06108 Nice cedex 02- France } \email{cbernard@unice.fr}

\author{ Patr\' icia Gon\c{c}alves}

\address{Patr\' icia Gon\c{c}alves\\ PUC-RIO\\ Departamento de Matem\'atica\\ Rua Marqu\^es de S\~ao Vicenteno. 225, 22453-900\\ G\'avea, Rio de Janeiro\\ Brazil and
CMAT\\ Centro de Matem\'atica da Universidade do Minho\\ Campus de Gualtar\\ 4710-057 Braga\\ Portugal}
\email{patricia@mat.puc-rio.br}

\author{Milton Jara}
\address{ Milton Jara\\ IMPA\\ Estrada Dona Castorina 110\\ Jardim Bot\^anico\\ CEP 22460-340\\ Rio de Janeiro\\ Brazil}
\email{mjara@impa.br}

\begin{abstract}
We consider a harmonic chain perturbed by an energy conserving noise and show that after a space-time rescaling the energy-energy correlation function is given by the solution of a skew-fractional heat equation with exponent $3/4$.
\end{abstract}

\maketitle

\section{Introduction}

The problem of anomalous diffusion of energy in one-dimensional chains of coupled oscillators has attracted a lot interest since the end of the 90's , see the review papers~\cite{D,LLP}. In one dimension the presence of long time tails in the correlation functions of the energy current  shows that transport coefficients are ill defined. Recently, following \cite{VB} {\footnote{In \cite{VB} the focus is on one-dimensional fluids.}}, Spohn goes further and gives in \cite{Sp} very precise predictions about the long-time behavior of the dynamic correlations of the conserved fields, identifying explicitly several universality classes. The predictions are based on the so-called nonlinear fluctuating hydrodynamics which claims that in order to capture the super-diffusive behavior of the system it suffices to expand the system of Euler equations up to the second order and add conservative space-time white noise satisfying the fluctuation-dissipation relation. These mesoscopic equations are the starting point from which the predictions are deduced. Thus, they do not depend on the microscopic specificities of the model but only on its behavior in a coarse space-time scale. The method applies also to conservative systems whose hydrodynamic equations are described by a system of hyperbolic conservation laws.

Up to now, mathematical progress on this issue has been rather modest. The validity of the hydrodynamic equations should be the consequence of good mixing properties of the microscopic dynamics, properties well known to be very difficult to justify rigorously for Hamiltonian systems. Therefore, during the last years, following the pioneering works {\cite{OVY}} and \cite{FFL}, it has been proposed to superpose stochastic perturbations to the deterministic Hamiltonian evolution in order to ensure the required chaoticity.  In \cite{BBO2} it is  proved that the thermal conductivity of an unpinned one-dimensional harmonic chain of oscillators perturbed by an energy-momentum conservative noise is infinite, while if a pinning potential (destroying momentum conservation) is added, it is finite. In \cite{BOS}, it is then shown that if the intensity $\ve$ of the noise goes to $0$, the local spectral density evolves according to a linear phonon Boltzmann equation in a space-time scale of order $\ve^{-1}$. The latter can be interpreted as the evolution of the density of a Markov process. In \cite{JKO}, \cite{JK}, the authors study the long time behavior of additive functionals of this Markov process and deduce that the long-time, large-scale limit of the solution of the previous Boltzmann equation converges to the solution of the fractional heat equation:
\begin{equation}
\partial_t u =- (-\Delta)^{3/4} u
\end{equation}
where $\Delta$ is the one-dimensional Laplacian (see also \cite{DLLMM}, \cite{DLLP}, \cite{LMMP} and references therein). This result is in perfect agreement with the nonlinear fluctuating hydrodynamics predictions (\cite{Sp}).  Nevertheless, observe that it is obtained in a {\textit{double limit}} procedure and that it is a priori much more difficult and interesting to obtain the fractional heat equation in a {\textit{unique}} space-time scaling limit bypassing the mesoscopic Boltzmann equation. The aim of this paper is to present a general method permitting, precisely, to solve this problem.

The model we consider in this paper has been introduced in \cite{BerSto} and presents strong analogies with the models described above. We believe, in fact, that with some extra work, the proof can be carried out also for the models of \cite{BBO2}. The systems of \cite{BerSto} can be described as follows. Let $V$ and $U$ be two non-negative potentials on $\RR$ and consider the Hamiltonian system $( \, {\bf r} (t) , {\bf p} (t) \,)_{t \ge 0}$, whose equations of motion are given by
\begin{equation}
\label{eq:generaldynamics}
\frac{dp_x}{dt} = V'(r_{x+1}) -V'(r_x),
\qquad \frac{dr_x}{dt} = U' (p_x) -U' (p_{x-1}),
\qquad x \in \ZZ,
\end{equation}
where $p_x$ is the momentum of the particle $x$, $q_x$ its position and $r_x=q_{x} -q_{x-1}$ is the ``deformation'' of the lattice at $x$. Standard chains of oscillators are recovered for a quadratic kinetic energy $U(p)=p^2 /2$. Now, take $V=U$, and call $\eta_{2x-1}=r_x$ and $\eta_{2x}=p_x$. The dynamics can be rewritten as:
\begin{equation}
\label{eq:dyneq}
d\eta_{x} (t) =\Big(V' (\eta_{x+1}) - V' (\eta_{x-1})\Big) dt.
\end{equation}
Notice that with these new variables, the energy of the system is simply given by $\sum_{x\in \bb Z} V(\eta_x)$. If $V(\eta)=\eta^2/2$, which is the case considered in this paper, then we recover a chain of harmonic oscillators. Then, following the spirit of \cite{BBO2}, the deterministic evolution is perturbed by adding a noise which consists to exchange $\eta_{x}$ with $\eta_{x+1}$ at random exponential times, independently for each bond $\{x,x+1\}$. The dynamics still conserves the energy $\sum_{x\in \ZZ} V(\eta_x)$ and the ``volume'' $\sum_{x\in \ZZ} \eta_x=\sum_{x \in \ZZ} [p_x +r_x]$ and destroys all the other conserved quantities. As argued in \cite{BerSto}, the volume conservation law is responsible for the anomalous energy diffusion observed for this class of energy-volume conserving dynamics. This can be shown for quadratic interactions (\cite{BerSto}) with a behavior similar to the one observed in \cite{BBO2} but also for exponential interactions (\cite{BG}). The technical advantage to deal with this kind of stochastic perturbation is that the number of conserved quantities is only $2$ (energy and volume) and not $3$ (energy, momentum and stretch) as it is for the dynamics of \cite{BBO2}. In a recent paper, Jara et al. (\cite{JKO2}) obtained similar results to ours, but by very different techniques,  for the dynamics of \cite{BBO2}. 

Our proof is based on some recent ideas introduced in \cite{J2}. One way to study the diffusivity of a conserved quantity of given system, is to look at the evolution of the space-time correlations of the conserved quantity on a diffusive (or $1:2:4$) space-time scaling, with respect to a given stationary state. For diffusive systems, these correlations evolve according to a linear heat equation, and the corresponding diffusion coefficient is what we call the diffusivity of the quantity at the given stationary state.

As we will see for the model described above, energy correlations evolve on a $1:2:3$ {\em superdiffusive} space-time scale. If we scale space with a mesh $\frac{1}{n}$, then we have to speed up the time by a factor $n^{3/2}$ in order to see a non-trivial evolution of the energy correlations. For the expert reader, we can explain why is it difficult to obtain a limiting evolution in this situation. Since the model we are looking at is conservative, the continuity equation relating spatial variations of the energy with the energy current, allows to perform an integration by parts which absorbs a factor $n$ of the time scale. If the system satisfies the, so-called, {\em gradient condition}, the Fourier's law is satisfies at the microscopic level, and the ergodic properties of the underlying dynamics are enough to perform a second integration by parts, absorbing an extra  factor $n$ of the time scale. This second integration by parts allows to obtain the heat equation as the limit of the correlations of the conserved quantity. If the system does not satisfy the gradient condition, the so-called {\em non-gradient method} introduced by Varadhan \cite{Var} allows to use a central limit theorem in order to show an approximate version of the {\em fluctuation-dissipation relation}, which allows to perform the second integration by parts. The non-gradient method is extremely technical and difficult to apply and it gives rigorous justification to the {\em Green-Kubo formula} for the diffusivity of a system.

If we believe that our scaling is the right one, what we need to perform is a sort of {\em fractional} integration by parts, since the extra factor $n^{1/2}$ would be overcome by a standard integration by parts. In \cite{J2} we introduced what we call the {\em quadratic correlation field} associated to the volume. This field has two different meaningful scaling limits. In the {\em hyperbolic} scaling $tn$, the volume correlations evolve according to a linear transport equation. In particular, the correlations do not evolve on a reference frame moving with the characteristic speed. In the {\em diffusive} time scaling $tn^2$ and on the same moving reference frame, the volume correlations follow the heat equation. It turns out that the energy current can be expressed as a {\em singular} functional of the quadratic correlation field. A two-dimensional Laplace problem can be used to express this singular functional in terms of a regular function of the quadratic field and a boundary term. This boundary term turns out to be a skew version of the fractional Laplacian of order $3/4$ of the energy, and in particular it allows to perform  a sort of fractional integration by parts.

%
%
%

The paper is organized as follows. In Section \ref{sec:model-results} we define the model and state the main result. In Section \ref{sec:wf} we give a formal intuitive proof, that is rigorously performed in Section \ref{sec:proof}.

\section{The model}
\label{sec:model-results}
\subsection{Description of the model}

For $\eta: \bb Z \to \bb R$ and $\alpha >0$, define
\begin{equation}
\tnorm \eta \tnorm_{\alpha} = \sum_{x \in \bb Z} \big| \eta(x) \big| e^{-\alpha |x|}.
\end{equation}
Define $\Omega_\alpha = \{ \eta: \bb Z \to \bb R; \tnorm \eta \tnorm_\alpha < +\infty\}$. The normed space $(\Omega_\alpha, \tnorm \cdot \tnorm)$ turns out to be a Banach space. In $\Omega_\alpha$ we consider the system of ODE's
\begin{equation}
\label{ODE}
\tfrac{d}{dt} \tilde{\eta}_t(x) = {\tilde \eta}_t(x+1) - \tilde{\eta}_t(x-1) \text{ for } t \geq 0 \text{ and } x \in \bb Z.
\end{equation}
The Picard-Lindel\"of Theorem shows that the system \eqref{ODE} is well posed in $\Omega_\alpha$. We will superpose to this deterministic dynamics a stochastic dynamics as follows. To each bond $\{x,x+1\}$, with $x \in \bb Z$ we associate an exponential clock of rate one. Those clocks are independent among them. Each time the clock associated to $\{x,x+1\}$ rings, we exchange the values of $\tilde \eta_t(x)$ and $\tilde \eta_t(x+1)$. Since there is an infinite number of such clocks, the existence of this dynamics needs to be justified. If we freeze the clocks associated to bonds not contained in $\{-M,\dots,M\}$, the dynamics is easy to define, since it corresponds to a piecewise deterministic Markov process. It can be shown that for an initial data $\eta_0$ in
\begin{equation}
\Omega = \bigcap_{\alpha>0} \Omega_\alpha,
\end{equation}
these piecewise deterministic processes stay at $\Omega$ and they converge to a well-defined Markov process $\{\eta_t; t\geq 0\}$, as $M \to \infty$, see \cite{BerSto} and the references therein. This Markov process is the rigorous version of the dynamics described above. Notice that $\Omega$ is a complete metric space with respect to the distance
\begin{equation}
d(\eta,\xi) = \sum_{\ell \in \bb N} \frac{1}{2^\ell} \min\{ 1, \tnorm \eta-\xi\tnorm_{\frac{1}{\ell}}\}.
\end{equation}
Let us describe the generator of the process $\{\eta_t; t\geq 0\}$. For $x,y \in \bb Z$ and $\eta \in \Omega$ we define $\eta^{x,y} \in \Omega$ as
\begin{equation}
\eta^{x,y}(z)=
\begin{cases}
\eta(y); & z=x\\
\eta(x); &z=y\\
\eta(z); &z \neq x,y.
\end{cases}
\end{equation}
We say that a function $f: \Omega \to \bb R$ is {\em local } if there exists a finite set $B \subseteq \bb Z$ such that $f(\eta) = f(\xi)$ whenever $\eta(x) = \xi(x)$ for any $x \in B$. For a smooth function $f: \Omega \to \bb R$ we denote by $\partial_x f: \Omega \to \bb R$ its partial derivative with respect to $\eta(x)$. For a function $f: \Omega \to \bb R$ that is local, smooth and bounded we define $L f: \Omega \to \bb R$ as $L f = S f + A f$, where
\begin{equation}
S f(\eta)  = \sum_{x \in \bb Z} \big(f(\eta^{x,x+1}) - f(\eta)\big),
\end{equation}
\begin{equation}
Af(\eta) = \sum_{x \in \bb Z} \big( \eta (x+1) -\eta (x-1) \big)\, \partial_x f (\eta)
\end{equation}
for any $\eta \in \Omega$. Denote by $\mc C_b(\Omega)$ the space of bounded functions $f:\Omega \to \bb R$ which are continuous with respect to the distance $d(\cdot,\cdot)$.  The generator of $\{\eta_t; t\geq 0\}$ turns out to be the closure in $\mc C_b(\Omega)$ of the operator $L$.

The process $\{\eta_t; t\geq 0 \}$ has a family $\{\mu_{\rho,\beta}; \rho \in \bb R, \beta >0\}$ of invariant measures given by
\begin{equation}
\mu_{\rho,\beta}(d\eta) = \prod_{x \in \bb Z} \sqrt{\tfrac{\beta}{2 \pi}} \exp\big\{ - \tfrac{\beta}{2}\big(\eta(x) - \rho \big)^2 \big\} d\eta(x).
\end{equation}
It also has two conserved quantities. If one of the numbers
\begin{equation}
\sum_{x \in \bb Z} \eta_0(x) , \quad \sum_{x \in \bb Z} \eta_0(x)^2
\end{equation}
is finite, then its value is preserved by the evolution of $\{\eta_t; t \geq 0\}$. Following \cite{BerSto}, we will call these conserved quantities {\em volume} and {\em energy}. Notice that $\int \eta(x) d\mu_{\rho,\beta} = \rho$ and $\int\eta(x)^2 d\mu_{\rho,\beta} = \rho^2 + \frac{1}{\beta}$.

\subsection{Description of the result}
Fix $\rho \in \bb R$ and $\beta>0$, and consider the process $\{\eta_t; t \geq 0\}$ with initial distribution $\mu_{\rho,\beta}$. Notice that $\{\eta_t+\lambda; t \geq 0\}$ has the same distribution of the process with initial measure $\mu_{\rho+\lambda,\beta}$. Therefore, we can assume, without loss of generality, that $\rho = 0$. We will write $\mu_\beta= \mu_{0,\beta}$ and we will denote by $E_\beta$ the expectation with respect to $\mu_\beta$. We will denote by $\bb P$ the law of $\{\eta_t; t \geq 0\}$ and by $\bb E$ the expectation with respect to $\bb P$. The {\em energy correlation function} $\{S_t(x); x \in \bb Z, t \geq 0\}$ is defined as
\begin{equation}
S_t(x) = \tfrac{\beta^2}{2} \,  \bb E\big[\big(\eta_0(0)^2-\tfrac{1}{\beta} \big) \big( \eta_t(x)^2 -\tfrac{1}{\beta} \big)\big]
\end{equation}
for any $x \in \bb Z$ and any $t \geq 0$. The constant $\frac{\beta^2}{2}$ is just the inverse of the variance of $\eta(x)^2-\frac{1}{\beta}$ under $\mu_\beta$. By translation invariance of the dynamics and the initial distribution $\mu_\beta$, we see that
\begin{equation}
\tfrac{\beta^2}{2}\,  \bb E\big[ \big(\eta_0(x)^2 -\tfrac{1}{\beta} \big) \big(\eta_t(y)^2 -\tfrac{1}{\beta}\big)\big]=S_t(y-x)
\end{equation}
for any $x, y \in \bb Z$. Our main result is the following

\begin{theorem}
\label{t1}
Let $f,g: \bb R \to \bb R$ be smooth functions of compact support. Then,
\begin{equation}
\label{ec1.12}
\lim_{n \to \infty} \tfrac{1}{n}\sum_{x, y  \in \bb Z} f\big( \tfrac{x}{n} \big) g\big( \tfrac{y}{n}\big) S_{tn^{3/2}}(x-y) = \iint f(x)g(y) P_t(x-y) dx dy,
\end{equation}
where $\{P_t(x); x \in \bb R, t \geq 0\}$ is the fundamental solution of the fractional heat equation
\begin{equation}
\label{ec1.13}
\partial_t u =  -\tfrac{1}{\sqrt{2}}\big\{ (-\Delta)^{3/4} - \nabla (-\Delta)^{1/4}\big\} u.
\end{equation}
\end{theorem}

A fundamental step in the proof of this theorem will be the analysis of the correlation function $\{S_t(x,y); x \neq y \in \bb Z, t \geq 0\}$ given by
\begin{equation}
S_t(x,y) = \tfrac{\beta^2}{2} \bb E \big[ \big(\eta_0(x)^2-\tfrac{1}{\beta} \big) \eta_t(x) \eta_t(y)\big]
\end{equation}
for any $t \geq 0$ and any $x \neq y \in \bb Z$. Notice that this definition makes perfect sense for $x =y$ and, in fact, we have $S_t(x,x) = S_t(x)$. For notational convenience we define $S_t(x,x)$ as equal to $S_t(x)$. However, these quantities are of different nature, since $S_t(x)$ is related to {\em energy fluctuations} and $S_t(x,y)$ is related to {\em volume fluctuations} (for $x \neq y$).

\begin{remark}
It is not difficult to see that with a bit of technical work our techniques actually show that the distributions valued process $\{ {\mc E}_t^n (\cdot) \, ; \, t\ge 0\}$ defined for any test function $f$ by
$${\mc E}_{t}^n (f) = \cfrac{1}{\sqrt n} \sum_{x \in \ZZ} f \big (\tfrac{x}{n} \big) ( \eta_{tn^{3/2}} (x)^2 - \tfrac{1}{\beta})$$
converges as $n$ goes to infinity to an infinite dimensional $3/4$-fractional Ornstein-Uhlenbeck process, i.e. the centered Gaussian process with covariance prescribed by the right hand side of (\ref{ec1.12}).
\end{remark}

\begin{remark}
It is interesting to notice that $P_t$ is the maximally asymmetric $3/2$-Levy distribution.
It has power law as $|x|^{-5/2}$ towards the diffusive peak and stretched exponential as $\exp[-|x|^3]$ towards the exterior of the sound cone (\cite{UZ} Chapter 4). As noticed to us by H. Spohn, this reflects the expected physical property that  no propagation beyond the sound cone occurs.
\end{remark}

In order to prove Theorem \ref{t1} we can assume $\beta=1$ since the general case can be recovered from this particular case by multiplying the process by $\beta^{-1/2}$. Thus, in the rest of the paper $\beta=1$.

\section{Duality}

Let $\bb H_2$ be the subspace of $L^2(\mu_\beta)$ spanned by the functions $\{\eta(x)\eta(y); x\neq y \in \bb Z\}$, $\{\eta(x)^2-\frac{1}{\beta}; x \in \bb Z\}$. As we can see on Appendix \ref{sec:Abb}, the space $\bb H_2$ is left invariant under the action of the operator $L$. By the definition of the generator of a Markov process, we know that for any bounded, local, smooth function $F: \Omega \to \bb R$,
\begin{equation}
\tfrac{d}{dt} \bb E[F(\eta_t)] = \bb E [ L F(\eta_t)]
\end{equation}
for any $t \geq 0$. Moreover, the Markov property shows that for any bounded function $G: \Omega \to \bb R$ and any $t \geq 0$,
\begin{equation}
\tfrac{d}{dt} \bb E[G(\eta_0)F(\eta_t)] = \bb E [G(\eta_0) L F(\eta_t)].
\end{equation}
Taking well-chosen approximating functions, we can show that these formulas hold for functions in $\bb H_2$. Using the fact that the operator $L$ leaves $\bb H_2$ invariant, we see that there exists an operator $\mc L: \ell^2(\bb Z^2) \to \ell^2(\bb Z^2)$ such that
\begin{equation}
\label{ec2.3}
\tfrac{d}{dt} S_t(x,y) = \mc L S_t(x,y)
\end{equation}
for any $t \geq 0$ and for any $x,y \in \bb Z$. In other words, the family of functions $\{S_t(x,y); t \geq 0; x,y \in \bb Z\}$ satisfy a {\em closed} set of equations. This property is known as {\em duality} in the literature, since it allows to solve \eqref{ec2.3} explicitly in terms of the semigroup associated to the operator $\mc L$. Therefore, in principle the analysis of scaling limits of the functions $\{S_t(x,y); t \geq 0; x,y \in \bb Z\}$ can be obtained as a consequence of the analysis of scaling limits of the operator $\mc L$. We will see that this approach is actually not convenient, because it misses the different roles played by the conserved quantities.\\

\section{Weak formulation of \eqref{ec2.3}}
\label{sec:wf}

Denote by $\mc C_c^\infty(\bb R)$ the space of infinitely differentiable functions $f: \bb R \to \bb R$ of compact support. Let $g \in \mc C_c^\infty(\bb R)$ be a fixed function.
For each $n \in \bb N$ we define the field $\{\mc S_t^n; t \geq 0\}$ as
\begin{equation}
\mc S_t^n(f) = \tfrac{1}{n}\!\! \sum_{x, y \in \bb Z} g\big(\tfrac{\vphantom{y}x}{n} \big) f\big( \tfrac{y}{n}\big) S_{tn^{3/2}}(y-x)
\end{equation}
for any $t \geq 0$ and any $f \in \mc C_c^\infty(\bb R)$. Rearranging terms in a convenient way we have that
\begin{equation}
\mc S_t^n(f) = \tfrac{\beta^2}{2} \bb E \Big[ \Big(\tfrac{1}{\sqrt n} \sum_{x \in \bb Z} g\big(\tfrac{x}{n}\big)\big(\eta_0(x)^2-\tfrac{1}{\beta}\big) \Big) \times
	\Big( \tfrac{1}{\sqrt n} \sum_{y \in \bb Z}f\big(\tfrac{y}{n}\big)  \big(\eta_{tn^{3/2}}(y)^2 -\tfrac{1}{\beta}\big)  \Big)\Big].
\end{equation}
For any function $f \in \mc C_c^\infty(\bb R)$, define the weighted $\ell^2(\bb Z)$-norm as
\begin{equation}
\|f\|_{2,n} = \sqrt{ \vphantom{H^H_H}\smash{\tfrac{1}{n} \sum_{x \in \bb Z} f\big( \tfrac{x}{n} \big)^2}}.
\end{equation}
\vspace{0pt}

\noindent
By the Cauchy-Schwarz inequality we have the {\em a priori} bound
\begin{equation}
\label{ec3.4}
\big| \mc S_t^n(f) \big| \leq \|g\|_{2,n} \| f\|_{2,n}
\end{equation}
for any $t \geq 0$, any $n \in \bb N$ and any $f,g \in \mc C_c^\infty(\bb R)$. Let $\mc C_c^\infty(\bb R^2)$ be the space of infinitely differentiable functions $h: \bb R^2 \to \bb R$. For a function $h \in \mc C_c^\infty(\bb R^2)$ we define $\{ Q_t^n(h); t \geq 0\}$ as
\begin{equation}
 Q_t^n(h) = \tfrac{\beta^2}{2} \bb E \Big[ \Big(\tfrac{1}{\sqrt n} \sum_{x \in \bb Z} g\big(\tfrac{x}{n}\big)\big(\eta_0(x)^2-\tfrac{1}{\beta}\big) \Big) \times
	\Big( \tfrac{1}{ n}\!\! \sum_{y \neq z  \in \bb Z}\!\! h\big(\tfrac{y}{n}, \tfrac{\vphantom{y}z}{n}\big) \eta_{tn^{3/2}}(y) \eta_{tn^{3/2}}(z)   \Big)\Big].
\end{equation}
In this way we have defined a two-dimensional field $\{Q_t^n; t \geq 0\}$. Notice that
\begin{equation}
Q_t^n(h) = \tfrac{1}{n^{3/2}} \sum_{x,y,z \in \bb Z} g\big(\tfrac{\vphantom{y}x}{n} \big) h \big(\tfrac{y}{n}, \tfrac{\vphantom{y} z}{n}\big) S_{t n^{3/2}}(y-x,z-x).
\end{equation}
Notice as well that $Q_t^n(h)$ depends only on the symmetric part of the function $h$. Therefore, we will always assume, without loss of generality, that $h(x,y) = h(y,x)$ for any $x,y \in \bb Z$. We point out that $Q_t^n(h)$ does not depend on the values of $h$ at the diagonal $\{x=y\}$. We have the {\em a priori} bound
\begin{equation}
\label{ec3.6}
\big| Q_t^n(h) \big| \leq 2 \|g\|_{2,n} \|{\tilde h} \|_{2,n},
\end{equation}
where $\|{\tilde h}\|_n$ is the weighted $\ell^2(\bb Z^2)$-norm of $\tilde h$:
\begin{equation}
\|{\tilde h}\|_{2,n} = \sqrt{  \vphantom{H^H_H}\smash{\tfrac{1}{n^2} \!\!\sum_{x, y \in \ZZ}\!\! {\tilde h} \big(\tfrac{\vphantom{y}x}{n} , \tfrac{y}{n}\big)^2}}
\end{equation}
and $\tilde h$ is defined by
\begin{equation*}
{\tilde h} \big(\tfrac{\vphantom{y}x}{n} , \tfrac{y}{n}\big) = h \big(\tfrac{\vphantom{y}x}{n} , \tfrac{y}{n}\big) \, {\bf 1}_{x \ne y}.
\end{equation*}
We notice that we use the same notation for the  weighted $\ell^2(\bb Z)$-norm for functions in $\mc C_c^\infty(\bb R)$ and $ \mc C_c^\infty(\bb R^2)$.
Using the computations of the Appendix \ref{sec:Abb}, we can obtain some differential equations satisfied by the fields $\mc S_t^n$ and $Q_t^n$. Before writing these equations down, we need to introduce some definitions. For a function $f \in \mc C_c^\infty(\bb R)$, we define  $\Delta_n f: \bb R \to \bb R$ as
\begin{equation}
\Delta_n f \big( \tfrac{x}{n}\big) = n^2\Big( f\big(\tfrac{x\plus 1}{n} \big) + f\big( \tfrac{x\minus 1}{n} \big) - 2 f\big(\tfrac{x}{n}\big) \Big).
\end{equation}
In other words, $\Delta_n f$ is a discrete approximation of the second derivative of $f$. We also define $\nabla_n f \otimes \delta : \sfrac{1}{n} \bb Z^2 \to \bb R$ as
\begin{equation}
\label{eq:3.9}
\big(\nabla_n f \otimes \delta\big) \big( \tfrac{x}{n}, \tfrac{y}{n}\big) =
\begin{cases}
\frac{n^2}{2}\big(f\big(\tfrac{x+1}{n}\big) - f\big(\tfrac{x}{n}\big)\big); & y =x\plus 1\\
\frac{n^2}{2}\big(f\big(\tfrac{x}{n}\big) - f\big(\tfrac{x-1}{n}\big)\big); & y =x\minus 1\\
0; & \text{ otherwise.}
\end{cases}
\end{equation}
Less evident than the interpretation of $\Delta_n f$, $\nabla_n f \otimes \delta$ turns out to be a discrete approximation of the distribution $f'\!(x) \otimes \delta(x=y)$, where $\delta(x=y)$ is the $\delta$ of Dirac at the line $x=y$. We have that
\begin{equation}
\label{ec3.10}
\tfrac{d}{dt} \mc S_t^n(f) = -2Q_t^n(\nabla_n f \otimes \delta) +  \mc S_t^n(\tfrac{1}{\sqrt n}\Delta_n f).
\end{equation}
In this equation we interpret the term $Q_t^n(\nabla_n f \otimes \delta)$ in the obvious way.
By the {\em a priori} bound \eqref{ec3.4}, the term $ \mc S_t^n(\frac{1}{\sqrt n} \Delta_n f)$ is negligible, as $n \to \infty$. If the scaling $tn^{3/2}$ is correct, the term $Q_t^n( \nabla_n f \otimes \delta)$ should be the relevant one. This motivates the study of the field $Q_t^n$. In order to describe the equation satisfied by $Q_t^n(h)$, we need some extra definitions. For $h \in \mc C_c^\infty(\bb R^2)$ we define $\Delta_n h : \bb R^2 \to \bb R$ as
\begin{equation}
\Delta_n h\big( \tfrac{\vphantom{y}x}{n}, \tfrac{y}{n}\big) = n^2\Big( h\big( \tfrac{\vphantom{y}x+1}{n}, \tfrac{y}{n}\big)+h\big( \tfrac{\vphantom{y}x-1}{n}, \tfrac{y}{n}\big)+h\big( \tfrac{\vphantom{y}x}{n}, \tfrac{y+1}{n}\big)+ h\big( \tfrac{\vphantom{y}x}{n}, \tfrac{y-1}{n}\big) - 4 h\big( \tfrac{\vphantom{y}x}{n}, \tfrac{y}{n}\big)\Big).
\end{equation}
In other words, $\Delta_n h$ is a discrete approximation of the Laplacian of $h$. We also define $\mc A_n h: \bb R \to \bb R$ as
\begin{equation}
\mc A_n h\big( \tfrac{\vphantom{y}x}{n}, \tfrac{y}{n}\big) = n\Big(h\big( \tfrac{\vphantom{y}x}{n}, \tfrac{y-1}{n}\big)+ h\big( \tfrac{\vphantom{y}x-1}{n}, \tfrac{y}{n}\big)- h\big( \tfrac{\vphantom{y}x}{n}, \tfrac{y+1}{n}\big)-h\big( \tfrac{\vphantom{y}x+1}{n}, \tfrac{y}{n}\big)\Big).
\end{equation}
The function $\mc A_n h$ is a discrete approximation of the directional derivative $(-2,-2) \cdot \nabla h$. Let us define $\mc D_n h : \sfrac{1}{n} \bb Z \to \bb R$ as
\begin{equation}
\mc D_n h\big( \tfrac{x}{n} \big) = n \Big( h \big(\tfrac{x}{n}, \tfrac{x+1}{n}\big) - h \big( \tfrac{x-1}{n}, \tfrac{x}{n} \big) \Big)
\end{equation}
and $\widetilde {\mc D}_n h :\sfrac{1}{n} \bb Z^2 \to \bb R$ as
\begin{equation}
\widetilde{\mc D}_n h (\tfrac{x}{n},\tfrac{y}{n}) =
\begin{cases}
n^2 \big(h\big(\tfrac{x}{n}, \tfrac{x+1}{n}\big)-h\big(\tfrac{x}{n}, \tfrac{x}{n}\big)\big); & y =x+1\\
n^2 \big(h\big(\tfrac{x-1}{n}, \tfrac{x}{n}\big)-h\big(\tfrac{x-1}{n}, \tfrac{x-1}{n}\big)\big); & y=x-1\\
0; & \text{ otherwise.}
\end{cases}
\end{equation}
The function $\mc D_n h$ is a discrete approximation of the directional derivative of $h$ along the diagonal $x=y$, while $\widetilde{\mc D}_n h$ is a discrete approximation of the distribution $\partial_y h(x,x) \otimes \delta(x=y)$.
Finally we can write down the equation satisfied by the field $Q_t^n(h)$:
\begin{equation}
\label{ec3.15}
\tfrac{d}{dt} Q_t^n(h) = Q_t^n\big(n^{-1/2}\Delta_n h+ n^{1/2}\mc A_n h\big) -2 \mc S_t^n\big( \mc D_n h\big) + 2 Q_t^n\big(n^{-1/2}\widetilde{\mc D}_n h\big).
\end{equation}

Notice that in equation \eqref{ec3.15}, both fields $\mc S_t^n$ and $Q_t^n$ appear with non-negligible terms. Moreover, the term involving $Q_t^n$ is quite singular, since it involves an approximation of a distribution. Looking at the equations \eqref{ec3.10} and \eqref{ec3.15} we see a possible strategy: given $f \in \mc C_c^\infty(\bb R)$, if we choose $h$ in a careful way, we can try to cancel out the terms $Q_t^n( \nabla_n f \otimes \delta)$ and $Q_t^n(n^{-1/2} \Delta_n h + n^{1/2} \mc A_n h)$. Then the term $\mc S_t^n(\mc D_n h)$ will provide a non-trivial drift for the differential equation \eqref{ec3.10} and with a little bit of luck the term $Q_t^n( n^{-1/2} \widetilde{\mc D}_n h)$ turns out to be negligible. This is the strategy that will be pursued in the following section.

\section{Proof of Theorem \ref{t1}}
\label{sec:proof}

In this section we prove Theorem \ref{t1}. We start with a non-rigorous discussion as a guideline of what are we going to do.

\subsection{Heuristics}

As explained above, the idea is to combine equations \eqref{ec3.10} and \eqref{ec3.15} in a clever way in order to obtain a weak formulation of a differential equation involving the field $\mc S_t^n$ alone. Let $h_n: \sfrac{1}{n} \bb Z^2 \to \bb R$ be the solution of the equation
\begin{equation}
\label{Poisson}
n^{-1/2} \Delta_n h \big(\tfrac{\vphantom{y}x}{n},\tfrac{y}{n}\big) + n^{1/2} \mc A_n h \big(\tfrac{\vphantom{y}x}{n},\tfrac{y}{n}\big)  =  \nabla_n f \otimes \delta \big(\tfrac{\vphantom{y}x}{n},\tfrac{y}{n}\big) .
\end{equation}
Define $\bb H = \{(x,y) \in \bb R^2; y \geq 0\}$.
It turns out that $h_n \big(\tfrac{x}{n}, \tfrac{y}{n} \big)$ is well approximated by
$ g\big(\tfrac{x+y}{2n},\tfrac{|x-y|}{2\sqrt n} \big)$, where $g: \bb H \to \bb R$ is the solution of the Laplace equation
\begin{equation}
\left\{
\begin{array}{cl}
\partial_y^2 g(x,y) - 4\partial_x g(x,y) =0 &\quad \text{   for } x \in \bb R, y > 0,\\
2\partial_y g(x,0) = f'(x) &\quad \text{ for } x \in \bb R.
\end{array}
\right.
\end{equation}
The solution $g$ of this problem is unique, regular and square-integrable. Therefore, we expect that
\begin{equation}
\lim_{n \to \infty} \sqrt n \|h_n\|^2_{2,n} = 2 \int_{\bb H} g(x,y)^2 dx dy.
\end{equation}
In particular, $\|h_n\|_{2,n} = \mc O(\frac{1}{n^{1/4}})$. We also expect that
\begin{equation}
\lim_{n \to\infty} \mc D_n h_n \big( \tfrac{x}{n}\big) = \partial_x g(x,0).
\end{equation}
Considering the integral formulation of the differential equation satisfied by the sum $\mc S_t^n(f) + 2Q_t^n(h_n)$, we see that
\begin{equation}
\label{ec4.6}
\mc S_t^n(f) = \mc S_0^n(f) -4 \int_0^t \mc S_s^n( \partial_x g(\cdot,0))ds +4 \int_0^t \tfrac{1}{\sqrt n} Q_s^n( {\tilde {\mc D}}_n h_n )ds
\end{equation}
plus terms of order $\mc O(\frac{1}{\sqrt n})$. At this heuristic level, we can argue that the second integral on the right-hand side of \eqref{ec4.6} is small, since it has a $\frac{1}{\sqrt n}$ in front of it. This is not straightforward and, in fact, replacing $h_n$ by the approximation furnished by the function $g$,  one observes that $n^{-{1/2}} {\tilde {\mc D}}_n h_n$ diverges with $n$. A more careful study of the true solution $h_n$ shows that  $n^{-1/2} {\tilde {\mc D}}_n h_n$ is, in fact, of order $1$ in $L^2$. But even with this estimate the a priori bound (\ref{ec3.6}) is not sufficient to show that $\int_0^t \tfrac{1}{\sqrt n} Q_s^n( {\tilde {\mc D}}_n h_n )ds$ is small. Some extra dynamical argument detailed in Subsection \ref{subsec:colp} proves that this term vanishes, as $n\to\infty$.

Using Fourier transform, it can be shown that
\begin{equation}
\partial_x g(\cdot,0) =\tfrac{1}{4\sqrt{2}} \left[ (-\Delta)^{3/4} - \nabla(-\Delta)^{1/4} \right] f.
\end{equation}
Therefore, \eqref{ec4.6} is an approximated weak formulation of \eqref{ec1.13}. With a little bit of work, we can show that for $f: [0,t] \times \bb R \to \bb R$ regular enough,
\begin{equation}
\mc S_t^n(f_t) = \mc S_0^n(f_0) + \int_0^t \mc S_s^n(\partial_t f_s + \bb L f_s) ds
\end{equation}
plus terms of order $\mc O(\frac{1}{\sqrt  n})$.
Here we have used the notation $\bb L = -\frac{1}{{\sqrt 2}}(-\Delta)^{3/4}-\frac{1}{\sqrt 2}\nabla(-\Delta)^{1/4}$. Passing to the limit and showing that the function
\begin{equation}
f_s(x) = \int f(y)P_{t-s}(y-x) dy
\end{equation}
can be used as a test function, Theorem \ref{t1} would be proved.

\subsection{Topology and relative compactness}

It is not straightforward to follow the strategy of proof of Theorem \ref{t1} outlined in the previous section. Therefore, we will divide the proof in various steps. For topological reasons it will be convenient to fix a finite time-horizon $T>0$.
In this section we start showing that the sequence $\{\mc S_t^n; t \in [0,T]\}_{n \in \bb N}$ is relatively compact. Of course, we need to specify the topology with respect to which this sequence is relatively compact. Let us define the {\em Hermite polynomials} $H_\ell: \bb R \to \bb R$ as
\begin{equation}
H_\ell(x) = (-1)^\ell e^{\frac{x^2}{2}} \frac{d^\ell}{dx^\ell} e^{-\frac{x^2}{2}}
\end{equation}
for any $\ell \in \bb N_0$ and any $x \in \bb R$. We define the {\em Hermite functions} $f_\ell: \bb R \to \bb R$ as
\begin{equation}
f_\ell(x) = \tfrac{1}{\sqrt{\ell! \sqrt{2 \pi}}} H_\ell(x) e^{-\frac{x^2}{4}}
\end{equation}
for any $\ell \in \bb N_0$ and any $x \in \bb R$. The Hermite functions $\{f_\ell; \ell \in \bb N_0\}$ form an orthonormal basis of $L^2(\bb R)$. For each $k \in \bb R$, we define the {\em Sobolev space} $\mc H_k$ as the completion of $\mc C_c^\infty(\bb R)$ with respect to the norm $\|\cdot\|_{\mc H_k}$ defined as
\begin{equation}
\|g\|_{\mc H_k} = \sqrt{\vphantom{H^H_H}\smash{\sum_{\ell \in \bb N_0} (1+\ell)^{2k} \<f_\ell,g\>^2 } }
\end{equation}

\vspace{8pt}
\noindent
for any $g \in \mc C_c^\infty(\bb R)$. Here we use the notation $\<f_\ell,g\> = \int g(x) f_\ell(x) dx$. Notice that $\mc H_0 = L^2(\bb R)$ and $\mc H_k \subseteq L^2(\bb R) \subseteq \mc H_{-k}$, for any $k >0$. By continuity, the inner product $\<\cdot,\cdot\>$ can be extended to a continuous bilinear form in $\mc H_{k} \times \mc H_{-k}$ for any $k >0$.
This bilinear form allows us to identify, for any $k \in \bb R$, the space $\mc H_{-k}$ with the dual of $\mc H_{k}$.
An important property is that the inclusion $\mc H_k \subseteq \mc H_{k'}$ is compact and Hilbert-Schmidt, whenever $k-k'>\frac{1}{2}$. The space $\mc H_{k}$ is a Hilbert space with respect to the inner product
\begin{equation}
\<g,h\>_k = \sum_{\ell \in \bb N_0} (1+\ell)^{2k} \<f_\ell,g\>\<f_\ell,h\>.
\end{equation}
Let us denote by $\mc C([0,T]; \mc H_{k})$ the space of continuous functions from $[0,T]$ to $\mc H_k$. We have the following compactness criterion in $\mc C([0,T]; \mc H_{k})$ for $k$ negative enough:

\begin{proposition}
\label{p1}
For any $k<-\frac{1}{2}$, a sequence $\{ S^n_t; t \in [0,T]\}_{n \in \bb N}$ of elements in the space $\mc C([0,T];\mc H_{k})$ is relatively compact if:
\begin{itemize}
\item[i)] for any $\ell \in \bb N_0$ the sequence of real-valued functions $\{\<S_t^n, f_\ell\>; t \in [0,T]\}_{n \in \bb N}$ is equicontinuous,
\item[ii)] the set $\{S_t^n(f_\ell); t \in [0,T]; n \in \bb N; \ell \in \bb N_0\}$ is bounded in $\bb R$.
\end{itemize}
\end{proposition}
\begin{proof}
By the Arzela-Ascoli theorem, we need to prove equicontinuity and boundedness of $\{ S^n_t; t \in [0,T]\}_{n \in \bb N}$ in $\mc C([0,T]; \mc H_k)$.
Notice that
\begin{equation}
\begin{split}
\sup_{|t-s| <\delta} \|S_t^n-S_s^n\|_{\mc H_k}^2
	&= \sup_{|t-s|<\delta} \sum_{\ell \geq 0} (1+\ell)^{2k} \big| \<S_t^n,f_\ell\> - \<S_s^n,f_\ell\>\big|^2\\
	&\leq \sum_{\ell \geq 0} \sup_{|t-s|<\delta}  (1+\ell)^{2k} \big| \<S_t^n,f_\ell\> - \<S_s^n,f_\ell\>\big|^2.\\
\end{split}
\end{equation}
Therefore, for each $M \in \bb N$,
\begin{equation}
\begin{split}
\sup_{|t-s| <\delta} \|S_t^n-S_s^n\|_{\mc H_k}^2
	&\leq \sum_{\ell=0}^{M-1} (1+\ell)^{2k}  \sup_{|t-s|<\delta} \big| \<S_t^n,f_\ell\> - \<S_s^n,f_\ell\>\big|^2\\
	&\quad \quad + 4\sup_{\substack{t \in [0,T]\\n \in \bb N\\ \ell \in \bb N_0}}\big|\<S_t^n,f_\ell\>\big|^2 \sum_{\ell \geq M} (1+\ell)^{2k}.
\end{split}
\end{equation}
By ii), making $M$ large enough and independent of $n$ or $T$ the second sum can be made arbitrarily small. Now that $M$ is fixed, the first sum can be made arbitrarily small taking $\delta$ small enough, independently of $n$ or $T$. This proves the equicontinuity of the sequence $\{S_t^n; t \in [0,T]\}_{n \in \bb N}$. The boundedness follows from ii) and a similar argument.
\end{proof}
Another very useful compactness criterion is given by the Banach-Alaoglu theorem, on its version for Hilbert spaces:

\begin{proposition}[Banach-Alaoglu theorem]
\label{pBA}
Let $\mc H$ be a separable Hilbert space. Any set $ \mc K \subseteq \mc H$ that is bounded with respect to the strong topology of $\mc H$ is sequentially, weakly relatively compact in $\mc H$.
\end{proposition}

We will use this proposition for the Hilbert spaces $\mc H_{-k}$ and $L^2([0,T]; \mc H_{-k})$ for $k$ big enough.

Recall the {\em a priori} bound \eqref{ec3.4}.
In order to make an effective use of Proposition \ref{pBA}, we need a way to estimate the $\ell^2_n(\bb Z)$-norm of various discretizations of $f_\ell$ in terms of their continuous counterparts. Let us denote by $\|\cdot\|_p$ the $L^p(\bb R)$-norm. We have the following lemma:

\begin{lemma}
\label{l1}
For any smooth function $f: \bb R \to \bb R$,
\begin{equation}
\Big| \tfrac{1}{n} \sum_{x \in \bb Z} f \big(\tfrac{x}{n} \big)^2 - \int f(x)^2 dx\Big| \leq \tfrac{2}{n}\|f^{\prime}\|_1  \|f\|_\infty .
\end{equation}
\end{lemma}
\begin{proof}
It is enough to observe that for any $a<b$,
\begin{equation}
\Big|\int_a^b \big(f(x)^2- f(a)^2\big) dx \Big| \leq 2(b-a) \sup_x |f(x)| \int_a^b |f'(x)|dx.
\end{equation}
\end{proof}

In view of this lemma, we need a way to compute $L^p(\bb R)$-norms of Hermite functions. We have the following:

\begin{proposition}
\label{p2}
For any $\delta >0$ there exists a constant $c=c(1,\delta)$ such that
\begin{equation}
\label{ec5.17}
\|f_\ell\|_1 \leq c(1+\ell)^{\frac{1+\delta}{4}}
\end{equation}
for any $\ell \in \bb N_0$. There also exists a constant $c(\infty)$ such that
\begin{equation}
\label{ec5.18}
\|f_\ell\|_\infty \leq \frac{c(\infty)}{(1+\ell)^{\frac{1}{6}}}.
\end{equation}
for any $\ell \in \bb N_0$.
\end{proposition}

The estimate \eqref{ec5.17} is proved in Appendix \ref{A.7}, and the estimate \eqref{ec5.18} is proved in \cite{KT} for example. Notice that any polynomial bound (even positive!) would have sufficed for what follows. Let us see how to use this proposition in order to obtain bounds on the $L^p$-norms of Hermite functions. The Hermite functions $\{f_\ell; \ell \in \bb N_0\}$ satisfy the relation
\begin{equation}
\label{ec4.20}
f_\ell' = \tfrac{1}{2} \big( \sqrt{\ell} f_{\ell-1} - \sqrt{\ell+1} f_{\ell+1}\big).
\end{equation}
Therefore, for any $\delta>0$ there exists a constant $c$ such that
\begin{equation}
\label{ec4.21}
\|f_\ell'\|_1 \leq c (1+\ell)^{3/4+\delta}
\end{equation}
for each $\ell \in \bb N_0$.
In particular, by Lemma \ref{l1} there exists a constant $c$ such that
\begin{equation}
\label{ec4.22a}
\|f_\ell\|_{2,n}^2 \leq 1 + \frac{ c(1+\ell)^{\frac{7}{12}+\delta}}{n}
\end{equation}
for any $\ell \in \bb N_0$ and any $n \in \bb N$. This estimate combined with the {\em a priori} bound \eqref{ec3.4} gives that
\begin{multline}
\label{ec4.23}
\|\mc S_t^n\|_{\mc H_{-k}}^2
	= \sum_{\ell \geq 0} (1+\ell)^{-2k} \big|\mc S_t^n(f_\ell)\big|^2\leq \\
	\leq \|g\|_{2,n}^2 \sum_{\ell \geq 0} \frac{1}{(1+\ell)^{2k}}\Big(1+\frac{c (1+\ell)^{\frac{7}{12}+\delta}}{n}\Big)
\end{multline}
for any $t \geq 0$. Since $g$ is smooth, by Lemma \ref{l1} $\|g\|_{2,n}$ is bounded in $n$. Therefore we conclude that
\begin{lemma}
\label{l2}
For any $k>\frac{19}{24}$, the sequence $\{\mc S^n_t; t \in [0,T]\}_{n \in \bb N}$ is sequentially, weakly relatively compact in $L^2([0,T]; \mc H_{-k})$. Moreover, for any $t \in [0,T]$ fixed, the sequence $\{\mc S_t^n; n \in \bb N\}$ is sequentially, weakly relatively compact in $\mc H_{-k}$.
\end{lemma}

\subsection{Characterization of limit points}
\label{subsec:colp}
In this section we obtain various properties satisfied by any limit point of $\{\mc S_t^n; t \in [0,T]\}_{n \in \bb N}$ and we will show that these properties characterize the limit point in a unique way.
Fix $k >\frac{19}{24}$ and let $\{\mc S_t; t \in [0,T]\}$ be a limit point of $\{\mc S_t^n; t \in [0,T]\}_{n \in \bb N}$ with respect to the weak topology of $L^2([0,T]; \mc H_{-k})$. With some abuse of notation, we will denote by $n$ the subsequence for which $\{\mc S_t^n; t \in [0,T]\}_{n \in \bb N}$ converges to $\{\mc S_t; t \in [0,T]\}$. Without loss of generality, we can assume that the distribution $\{\mc S_t^n\}_{n \in \bb N}$ converges to $\mc S_t$ with respect to the weak topology of $\mc H_{-k}$ and that the path $\{\mc S_s^n; s \in [0,t]\}_{n \in \bb N}$ converges to $\{\mc S_s; s \in [0,t]\}$ with respect to the weak topology of $L^2([0,t]; \mc H_{-k})$ for any $t \in [0,T]$ such that $\frac{t}{T} \in \bb Q$.
In order to simplify the notation, we define $[0,T]_{\bb Q} = \{t \in [0,T]; \frac{t}{T} \in \bb Q\}$.

Fix a function $f \in \mc C_c^\infty(\bb R)$ and let $h_n : \frac{1}{n} \bb Z \times  \frac{1}{n} \bb Z \to \bb R$ be the solution of the equation
\begin{equation}
\label{ec4.22}
n^{-1/2}\Delta_n h+ n^{1/2}\mc A_n h= \nabla_n f \otimes \delta.
\end{equation}
The following properties of $h_n$ are shown in the Appendix \ref{sec:proof_l3}:

\begin{lemma}
\label{l3}
Let $f \in \mc C_c^\infty(\bb R)$.
The solution of \eqref{ec4.22} satisfies
\begin{equation}
\label{eq:l31}
\lim_{n \to \infty} \frac{1}{n^2} \sum_{x,y \in \bb Z} h_n\big(\tfrac{\vphantom{y}x}{n}, \tfrac{y}{n}\big)^2 =0
\end{equation}
and
\begin{equation}
\label{eq:l32}
\lim_{n \to \infty} \frac{1}{n} \sum_{x \in \bb Z} \big|\mc D_n h_n \big(\tfrac{x}{n}\big) + \tfrac{1}{4} \bb L f\big(\tfrac{x}{n}\big) \big|^2 =0.
\end{equation}
In other words, $\|h_n\|_{2,n}$, $\|\mc D_n h_n + \tfrac{1}{4} \bb Lf \|_{2,n}$ converge to $0$, as $n \to \infty$.
\end{lemma}
By \eqref{ec3.10} and \eqref{ec3.15}, we see that
\begin{equation}
\begin{split}
\mc S_T^n(f) &= \mc S_0^n(f) + \int_0^T \mc S_t^n\big(-4\mc D_n h_n\big) dt+ 2 \big[Q_0^n(h_n) - Q_T^n(h_n)\big] \\
&\hspace{3cm}+ \int_0^T \mc S_t^n\big(\tfrac{1}{\sqrt n}\Delta_n f \big) dt + 4\int_0^T {Q}_t^n \big( \tfrac{1}{\sqrt n}\widetilde{\mc D}_n(h_n) \big) dt.
\end{split}
\end{equation}
Therefore, by the {\em a priori} bound \eqref{ec3.6} and by Lemma \ref{l3}, we have that
\begin{equation}
\label{ec4.277}
\mc S_T^n(f) = \mc S_0^n(f) + \int_0^T \mc S_t^n(\bb L f) dt + 4\int_0^T {Q}_t^n \big( \tfrac{1}{\sqrt n}\widetilde{\mc D}_n(h_n) \big) dt
\end{equation}
plus an error term which goes to $0$, as $n \to \infty$. As explained above, it turns out that the {\em a priori} bound \eqref{ec3.6} is not sufficient to show that the last term on the right hand side of  (\ref{ec4.277}) goes to $0$, as $n\to\infty$, since
\begin{equation}
\label{eq:l33}
\frac{1}{n^3} \sum_{x \in \bb Z} \widetilde{\mc D}_n h_n \big(\tfrac{x}{n}, \tfrac{x+1}{n}\big)^2
\end{equation}
is of order one. Therefore, we use again (\ref{ec3.15}) applied to $h=v_n$ where $v_n$, is the solution of the Poisson equation
\begin{equation}
\label{eq:poisson-v}
n^{-1/2} \Delta_n v_n \big(\tfrac{\vphantom{y}x}{n},\tfrac{y}{n}\big) + n^{1/2} \mc A_n v_n\big(\tfrac{\vphantom{y}x}{n},\tfrac{y}{n}\big)  =n^{-1/2} {\widetilde {\mc D}}_n h_n.
\end{equation}
Then we have
\begin{equation*}
\begin{split}
\int_0^T Q_t^n \big( \tfrac{1}{\sqrt n} {\widetilde{\mc D}}_n h_n \big) dt &= 2 \int_0^T {\mc S}_t^n ({\mc D}_n v_n) dt -2 \int_{0}^T {Q}_t^n  \big( \tfrac{1}{\sqrt n} {\widetilde{\mc D}}_n v_n \big) dt \\
&+ Q_T^n (v_n) -Q_0^n (v_n).
\end{split}
\end{equation*}
Now, we use the {\em a priori} bounds \eqref{ec3.4} and \eqref{ec3.6}. We have the following estimates on $v_n$ which are proved in the Appendix \ref{subsec:secondterm}.
\begin{lemma}
\label{lem:5,7.}
The solution $v_n$ of \eqref{eq:poisson-v} satisfies
\begin{equation}
\label{est vn}
\lim_{n \to \infty} \; \tfrac{1}{n^2}\sum_{x,y\in\bb Z}v_n\big(\tfrac {\vphantom{y}x}{n},\tfrac{y}{n}\big)^2=0,
\end{equation}
\begin{equation}\label{est der vn}
\lim_{n \to \infty} \;\tfrac{1}{n}\sum_{x\in\bb Z}\mc D_n v_n\big(\tfrac {x}{n}\big)^2=0,
\end{equation}
\begin{equation}\label{est tilde d vn}
\lim_{n \to \infty} \;\tfrac{1}{n^3}\sum_{x\in\bb Z}\widetilde{\mc D}_n v_n\big(\tfrac {x}{n},\tfrac {x+1}{n}\big)^2=0.
\end{equation}
In other words,  $ \left\| v_n \right\|_{2,n}$, $\left\| {\mc D}_n  v_n \right\|_{2,n}$ and $\left\| \tfrac{1}{\sqrt n} {\widetilde {\mc D}}_n v_n \right\|_{2,n}$ converge to $0$, as $n\to\infty$.
\end{lemma}
It follows that
\begin{equation}
\label{ec4.27}
\mc S_T^n(f) = \mc S_0^n(f) +\int_0^T \mc S_t^n(\bb L f) dt
\end{equation}
plus an error term which goes to $0$, as $n \to \infty$. Recall that $\{\mc S_t^n; t \in [0,T]\}_{n \in \bb N}$ converges weakly to $\{\mc S_t; t \in [0,T]\}$. Therefore, we could take the limit in \eqref{ec4.27} if we could show that $\bb Lf \in \mc H_{k}$. It turns out that this is not the case. In fact, the operator $\bb L$ is an integro-differential operator with heavy tails. Even for $f \in \mc C_c^\infty(\bb R)$ the function $\bb L f$ has heavy tails. We can show the following:

\begin{lemma}[Lemma 2.8, \cite{DG}]
\label{l4}
For any $f \in \mc C_c^\infty(\bb R)$ there exists a constant $c=c(f)$ such that
\begin{equation}
\label{ec4.28}
\big| \bb L f(x) \big| \leq \frac{c}{(1+x^2)^{5/4}}
\end{equation}
for any $x \in \bb R$.
\end{lemma}

An important consequence of this lemma is that $\bb L f \in L^2(\bb R)$. Notice that $f'$ also satisfies the hypothesis of the lemma, and therefore we can take $c$ such that we also have
\begin{equation}
\label{ec4.31}
\big|\tfrac{d}{dx} \bb L f(x) \big| \leq \frac{c}{(1+x^2)^{5/4}}
\end{equation}
for any $x \in \bb R$. Using Lemma \ref{l1} we conclude that $\|\bb Lf\|_{2,n}$ is uniformly bounded in $n$. Moreover, it can be approximated by functions in $\mc H_{k}$, {\em uniformly} in $n$. In fact, consider the {\em bump function} $\phi: \bb R \to \bb R$  given by
\begin{equation}
\phi(x) = \int_{|x|}^\infty e^{-\frac{1}{y(1-y)}} \mathbf{1}_{\{y\in [0,1]\}} dy
\end{equation}
and define for $M \in \bb N$ the function $g_M:\bb R \to \bb R$ as
\begin{equation}
g_M(x) =
\begin{cases}
1,& |x| \leq M,\\
\frac{\phi(|x|-M)}{\phi(0)}, & |x| >M.
\end{cases}
\end{equation}
Using \eqref{ec4.28} and \eqref{ec4.31} we see that
\begin{equation}
\label{ec4.34}
\lim_{M \to \infty} \sup_{n \in \bb N} \|(1-g_M)\bb L f\|_{2,n} =0.
\end{equation}
We claim that $g_M \bb Lf \in \mc H_k$ for any $k >0$. Notice that $g_M \bb L f \in \mc C_c^\infty(\bb R)$. Therefore, this will be a consequence of the following
\begin{lemma}
\label{l5}
Let $f$ be a smooth function with $\lim_{x \to \pm \infty} f(x)=0$. Assume that $(\partial_x-\frac{x}{2})f \in L^2(\bb R)$. Then $f \in \mc H_k$ for any $k\leq \frac{1}{2}$. In particular, if $f \in \mc C_c^\infty(\bb R)$, then $f \in \mc H_{k}$ for any $k \in \bb R$.
\end{lemma}
\begin{proof} Using the relation $H_{\ell+1}' = (\ell+1) H_\ell$ and integrating by parts, we see that
\begin{equation}
\<f_\ell, f\> = \frac{-1}{\sqrt{\ell+1}} \< f_{\ell +1}, \big(\partial_x -\tfrac{x}{2}\big) f\>.
\end{equation}
Therefore,
\begin{equation}
\sum_{\ell \geq 0} (1+\ell)\<f_\ell,f\>^2 = \sum_{\ell \geq 1} \<f_\ell, \big( \partial_x -\tfrac{x}{2}\big) f\>^2 \leq \big\|\big(\partial_x-\tfrac{x}{2}\big)f\big\|_{2}^{2}\, ,
\end{equation}
which shows the first part of the lemma. Repeating the argument $j$ times, we see that
\begin{equation}
\sum_{\ell \geq 0} (1+\ell)^j \<f_\ell,f\>^2 \leq \big\| \big(\partial_x-\tfrac{x}{2}\big)^j f\big\|_2^2,
\end{equation}
which shows the second part of the lemma.
\end{proof}

Using \eqref{ec4.34} and \eqref{ec3.4} we can write \eqref{ec4.27} as
\begin{equation}
\mc S_T^n(f) = \mc S_0^n(f) + \int_0^T \mc S_t^n(g_M \bb L f) dt
\end{equation}
plus a rest that goes to $0$, as $n \to \infty$ {\em and then} $M \to \infty$.

  Now we can pass to the limit on each one of the terms in this equation, since $g_M \bb L f \in \mc H_k$. Taking $n \to \infty$ and then $M \to \infty$ we conclude that
\begin{equation}
\mc S_T(f) = \mc S_0(f) + \int_0^T \mc S_t({\bb L} f) dt
\end{equation}
for any $f \in \mc C_c^\infty(\bb R)$. Repeating the arguments above for $t \in [0,T]_{\bb Q}$ we see that
\begin{equation}
\label{ec4.40}
\mc S_t(f) = \mc S_0(f) + \int_0^t \mc S_s(\bb Lf) ds
\end{equation}
for any $t \in [0,T]_{\bb Q}$.
Notice that the {\em a priori} bound \eqref{ec3.4} is stable under weak limits, and therefore we have that
\begin{equation}
\big| \mc S_t(f) \big| \leq  \|g\|_{2} \|f\|_2
\end{equation}
for any $t \in [0,T]_{\bb Q}$ and any $f \in \mc C_c^\infty(\bb R)$. Using this bound back into \eqref{ec4.40}, we see that
\begin{equation}
\label{ec4.42}
\big|\mc S_t(f) - \mc S_s(f)\big| \leq |t-s| \|g\|_2 \|\bb L f\|_2.
\end{equation}
for any $s,t \in [0,T]_{\bb Q}$. In other words, the function $t \mapsto \mc S_t(f)$, defined for $t \in [0,T]_{\bb Q}$ is uniformly Lipschitz. In particular, it can be continuously extended to $[0,T]$ in a unique way. Here we face a problem: this extension does not need to be equal to $\{\mc S_t; t \geq 0\}$, since the latter is an element of $L^2([0,T]; \mc H_{-k})$. If we can prove that $\{\mc S_t(f); t \in [0,T]\}$ is continuous, then both processes would be equal. The idea is to use the compactness criterion of Proposition \ref{p1}. It turns out that it is not convenient to use this Lemma for the sequence $\{\mc S_t^n; t \in [0,T]\}_{n \in \bb N}$ but for another auxiliary sequence. Fix $n \in \bb N$ and let $\{\widetilde{\mc S}_t^n; t \in [0,T]\}$ be the field given by
\begin{equation}
\label{ec5.49}
\widetilde{\mc S}_t^n(f) = \mc S_0^n(f) + \int_0^t \mc S_s^n(\bb L f) ds
\end{equation}
for any $t \in [0,T]$ and any $f \in \mc C_c^\infty(\bb R)$. We assert  that the sequence $\{\widetilde{\mc S}_t^n(f); t \in [0,T]\}_{n \in \bb N}$ is relatively compact in $\mc C([0,T]; \mc H_{-k})$. According to Proposition \ref{p1}, we have to prove two properties, namely equicontinuity and uniform boundedness of $\{\widetilde{\mc S}_t^n(f_\ell); t \in [0,T]\}_{n \in \bb N}$ for each $\ell \in \bb N_0$. Boundedness follows at once from the {\em a priori} bound \eqref{ec3.4}. Looking at \eqref{ec5.49}, in order to show equicontinuity, it is enough to show that $\mc S_t^n(\bb L f_\ell)$ is uniformly bounded in $t$ and $n$. But this is again an easy consequence of the {\em a priori} bound \eqref{ec3.4} and the discussion after Lemma \ref{l4}. Therefore, the sequence $\{\widetilde{\mc S}_t^n; t \in [0,T]\}_{n\in\mathbb{N}}$ is relatively compact in $\mc C([0,T]; \mc H_{-k})$ for any $k > \frac{1}{2}$. In particular, it has at least one limit point $\{\widetilde{\mc S}_t; t \in [0,T]\}$. Since this topology is stronger than the topology of $L^2([0,T]; \mc H_{-k})$, this limit has to be $\{\mc S_t; t \in [0,T]\}$. Therefore, we have proved that $\{\mc S_t; t \in [0,T]\}$ is continuous.

Let $f: [0,T] \times \bb R \to \bb R$ be a smooth function of compact support (in $[0,T] \times \bb R$). The estimate \eqref{ec4.42} and the continuity of $\{\mc S_t; t \in [0,T]\}$ allows us to show the following extension of \eqref{ec4.40}:
\begin{equation}
\label{ec4.43}
\mc S_T(f_T) = \mc S_0(f_0) + \int_0^T \mc S_t\big((\partial_t+\bb L)f_t\big) dt.
\end{equation}
What \eqref{ec4.42} is saying is that $\{\mc S_t; t \in [0,T]\}$ is a {\em weak solution} of \eqref{ec1.13}, as defined in (2.1) of \cite{J}. In Section 8.1 of that paper, it is shown that there exists a unique solution of \eqref{ec4.43}. This uniqueness result shows that the limit process $\{\mc S_t; t \in[0,T]\}$ is unique. Now we are close to finish the proof of Theorem \ref{t1}.
In fact, we have shown that the sequence $\{\mc S_t^n; t \in [0,T]\}_{n \in \bb N}$ is relatively compact with respect to the weak topology in $L^2([0,T]; \mc H_{-k})$ for any $k >\frac{19}{24}$, and that this sequence has exactly one limit point. Therefore, the sequence $\{\mc S_t^n; t \in [0,T]\}_{n \in \bb N}$, actually, converges to that unique limit point, which we called $\{\mc S_t; t \in [0,T]\}$. The convergence also holds for any fixed time $t \in [0,T]_{\bb Q}$, with respect to the weak topology of $\mc H_{-k}$. Since $T$ is arbitrary, this last convergence holds for any $t \in [0,\infty)$.  In particular, $\mc S_t^n(f)$ converges to $\mc S_t(f)$, as $n \to \infty$, for any $f \in \mc C_c^\infty(\bb R)$. But this is exactly what \eqref{ec1.12} says. Therefore, Theorem \ref{t1} is proved.

\section*{Acknowledgements}
The authors thanks S. Olla and H. Spohn for their interest in this work and highly useful discussions.\\

P.G. thanks FCT (Portugal) for support through the research
project ``Non - Equilibrium Statistical Physics" PTDC/MAT/109844/2009 and to CNPq (Brazil) for support through the research project ``Additive functionals of particle systems", Universal n. 480431/2013-2.
P.G. thanks the Research Centre of Mathematics of the University of
Minho, for the financial support provided by ``FEDER" through the
``Programa Operacional Factores de Competitividade  COMPETE" and by
FCT through the research project  PEst-OE/MAT/UI0013/2014.
C.B. acknowledges the support of the French Ministry of
Education through the grant ANR-10-BLAN 0108 (SHEPI). C.B. and P.G. are grateful to \' Egide and FCT for support through the research project FCT/1560/25/1/2012/S. C.B. and M.J. are grateful to the "Brazilian-French Network in Mathematics".

\newpage

{\bf{Warning:}} In the sequel, we denote by $C, c, \ldots$ some positive constants. Sometimes, in order to precise that the constant $C$ depends specifically on a parameter $a$ we write $C(a)$. The constants can change from line to line and, nevertheless, be denoted by the same letter.

\appendix

\section{Computations involving the generator $L$}
\label{sec:Abb}

Let $f: \bb Z \to \bb R$ be a function of finite support, and let ${\mc E} (f) : \Omega \to \RR$ be defined as
\begin{equation*}
{\mc E} (f) =\sum_{{x \in \bb Z}} f(x)\eta(x)^2.
\end{equation*}
A simple computation shows that
\begin{equation*}
S{\mc E} (f) =\sum_{{x \in \bb Z}} \Delta f(x)\eta^2(x),
\end{equation*}
where $\Delta f(x)  =  f(x\!+\!1)+f(x\!-\!1) -2f(x)$ is the discrete Laplacian on $\bb Z$.
On the other hand
\begin{equation*}
A{\mc E} (f) =-2\sum_{{x \in \bb Z} } \nabla f(x)\eta(x)\eta(x\!+\!1),
\end{equation*}
where $\nabla f(x)  =  f(x\!+\!1) -f(x)$ is the discrete right-derivative in $\bb Z$.

Let $f: \bb Z^2 \to \bb R$ be a symmetric function of finite support, and let ${ Q} (f) : \Omega \to \RR$ be defined as
\vspace{-10pt}
\begin{equation*}
{Q} (f) =\sum_{\substack{x,y \in \bb Z \\ x \neq y}} \eta (x) \eta (y) f(x,y).
\end{equation*}
Define $\Delta f: \bb Z^2 \to \bb R$ as
\begin{equation}
\Delta f(x,y)  =  f(x\!+\!1,y)+f(x\!-\!1,y) + f(x,y\!+\!1)+f(x,y\!-\!1)-4f(x,y)
\end{equation}
for any $x,y \in \bb Z$ and $\mc A f: \bb Z^2 \to \bb R$ by
\begin{equation}
\mc Af(x,y) = f(x\!-\!1,y) + f(x,y\!-\!1) - f(x\!+\!1,y) - f(x,y\!+\!1)
\end{equation}
for any $x,y \in \bb Z$. Notice that $\Delta f$ is the discrete Laplacian on the lattice $\bb Z^2$ and $\mc A f$ is a possible definition of the discrete derivative of $f$ in the direction $(-2,-2)$. Notice that we are using the same symbol $\Delta$ for the one-dimensional and two-dimensional, discrete Laplacian. From the context it will be clear which operator we will be using. We have that
\begin{equation}
\begin{split}
S Q (f)
		&=\!\! \sum_{|x-y| \ge 2}\!\! f(x,y) \big[ \eta (y) \Delta \eta (x) + \eta (x) \Delta \eta (y) \big]\\
		&\quad + 2 \sum_{x\in\bb Z} f(x,x\!+\!1) \big[ (\eta (x\!-\!1) -\eta (x)) \eta (x\!+\!1) + \\
		&\quad \quad +(\eta (x\!+\!2) -\eta (x\!+\!1)) \eta (x)\big]\\
		&=\!\! \sum_{x,y\in\bb Z}\!\! \Delta f(x,y) \eta (x) \eta (y) -2 \sum_{x \in \bb Z} f(x,x) \eta (x) \Delta \eta(x) \\
		&\quad-2 \sum_{x\in\bb Z}  f(x,x\!+\!1) \big[ \eta (x\!+\!1) \Delta \eta(x) + \eta (x) \Delta \eta(x\!+\!1)\big]\\
		&\quad \quad + 2 \sum_{x\in\bb Z}  f(x,x\!+\!1) \big[ \eta (x\!+\!1) \eta (x\!-\!1) +\eta (x\!+\!2) \eta (x) -2 \eta (x) \eta (x\!+\!1)\big].\\
\end{split}
\end{equation}
Grouping terms involving $\eta(x)^2$ and $\eta(x)\eta(x\!+\!1)$ together we get that
\begin{equation}
\begin{split}
S Q(f)
		&= \sum_{\substack{x,y \in \bb Z \\ x \neq y}} ({\Delta} f)(x,y) \eta (x)\eta (y) \\
		&\quad + 2 \sum_{x \in \bb Z}  \Big\{ \big[f(x,x\!+\!1)-f(x,x)\big] +\\
		&\quad \quad +\big[f(x,x\!+\!1) -f(x\!+\!1, x\!+\!1)\big]\Big\} \eta (x)\eta(x\!+\!1)\\
		&= Q(\Delta f) + 2 \sum_{x\in\bb Z}  \Big\{\big[ f(x,x\!+\!1) -f(x,x)\big] +\\
		&\quad \quad +\big[f(x,x\!+\!1) -f(x\!+\!1, x\!+\!1)\big]\Big\} \eta (x)\eta(x\!+\!1)\\
\end{split}
\end{equation}
Similarly, we have that
\begin{equation}
\begin{split}
A{Q} (f) &= \!\!\sum_{\substack{x,y \in \bb Z \\ x \neq y}}\!\!  \mc A f(x,y) \eta(x) \eta(y)\\
		&\quad+ 2 \sum_{x \in \bb Z} \Big\{ \eta(x)^2 \big[ f(x\!-\!1,x) -f(x,x\!+\!1)\big] \\
		&\quad \quad  - \eta(x) \eta(x\!+\!1) \big[ f(x,x)-f(x\!+\!1,x\!+\!1)\big]\Big\}\\
&= Q(\mc A f)\\
		&\quad+ 2 \sum_{x \in \bb Z} \Big\{ \eta(x)^2 \big[ f(x\!-\!1,x) -f(x,x\!+\!1)\big] \\
		&\quad \quad- \eta(x) \eta(x\!+\!1) \big[ f(x,x)-f(x\!+\!1,x\!+\!1)\big]\Big\}.
\end{split}
\end{equation}
It follows that
\vspace{-10pt}
\begin{equation}
\label{ecA.8}
L Q(f) = Q((\Delta + \mc A)f) + D(f),
\end{equation}
where the diagonal term $D(f)$ is given by
\begin{equation}
\begin{split}
D(f)  &=2 \sum_{x \in \bb Z} \big(\eta(x)^2- \tfrac{1}{\beta}\big) \big(f(x\minus 1,x) - f(x,x \plus 1)\big)\\
	&\quad+ 4 \sum_{x \in \bb Z} \eta(x) \eta(x \plus 1) \big( f(x,x \plus 1) -f(x,x)\big).
\end{split}
\end{equation}
The normalization constant $\frac{1}{\beta}$ can be added for free because $f(x,x\plus 1) - f(x\minus 1,x)$ is a mean-zero function. The diagonal term will be of capital importance, in particular the term involving $\eta(x)^2$. Notice that the operators $f \mapsto Q(f)$, $f \mapsto L Q(f)$ are continuous maps from $\ell^2(\bb Z^2)$ to $L^2(\mu_\beta)$. Therefore, an approximation procedure shows that the identities above hold true for any $f \in \ell^2(\bb Z^2)$.

\section{Tools of Fourier analysis}

Let $d \ge 1$ and let $x \cdot y $ denote the usual scalar product in $\RR^d$ between $x$ and $y$. The Fourier transform  of a function $g :\tfrac{1}{n} \ZZ^d \to \RR$ is defined by
\begin{equation*}
{\widehat g}_n (k) = \tfrac{1}{n^d} \sum_{x \in \ZZ^d} g (\tfrac{x}{n}) e^{ \tfrac{2i \pi k \cdot x}{n}}, \quad k\in \RR^d.
\end{equation*}
The function $\widehat{g}_n$ is $n$-periodic in all the directions of $\RR^d$. We have the following Parseval-Plancherel identity  between the $\ell^2$-norm of $g$, weighted by the natural mesh, and the $L^2 ([-\tfrac{n}{2}, \tfrac{n}{2}]^d)$-norm of its Fourier transform:
\begin{equation}
\| g \|^2_{2,n} := \tfrac{1}{n^d} \sum_{x \in \ZZ^d} | g (\tfrac{x}{n})|^2 = \int_{[-\tfrac{n}{2}, \tfrac{n}{2}]^d}\;  \left| {\widehat g}_n (k) \right|^2 \, dk \; := \| {\widehat g_n} \|_2^2.
\end{equation}

The function $g$ can be recovered from the knowledge of its Fourier transform by the inverse Fourier transform of ${\widehat g}_n$:
\begin{equation}
g ( \tfrac{x}{n}) = \int_{[-\tfrac{n}{2}, \tfrac{n}{2}]^d}\; {\widehat g}_n (k)\;  e^{- \frac{2i \pi x\cdot k}{n}} \, dk.
\end{equation}

For any $p\ge 1$ let $[(\nabla_n)^p ]$ denote the $p$-th iteration of the operator $\nabla_n$.

\begin{lemma}
\label{lem:sfp}
Let $f: \tfrac{1}{n} \ZZ \to \RR$ and $p \ge 1$ be such that
\begin{equation}
\label{eq:asssfp}
\frac{1}{n} \sum_{x \in \ZZ} \left| [(\nabla_n)^p ] f\big( \tfrac{x}{n} \big) \right| < +\infty.
\end{equation}
There exists a universal constant $C:=C(p)$ independent of $f$ and $n$  such that for any $|y| \le 1/2$,
\begin{equation*}
| \widehat{f_n} (yn) | \le \frac{C}{n^p |\sin(\pi y)|^p} \; \left| \frac{1}{n} \sum_{x \in \ZZ} [(\nabla_n)^p ] f\big( \tfrac{x}{n} \big) e^{2i \pi y x} \right|.
\end{equation*}
In particular, if $f$ is in the Schwartz space  ${\mc S} (\RR)$ then for any $p\ge 1$, there exists a constant $C:=C(p,f)$ such that for any $|y| \le 1/2$,
\begin{equation*}
| \widehat{f_n} (yn) | \le \frac{C}{1+ (n|y|)^{p}}.
\end{equation*}

\end{lemma}

\begin{proof}
For the first claim it is sufficient to show that
\begin{equation}
\label{eq:dibp}
\frac{1}{n} \sum_{x \in \ZZ} f\big( \tfrac{x}{n} \big) e^{2i \pi y x} \; = \; - \frac{e^{i\pi y}}{2\, i\, n \sin (\pi y)} \; \frac{1}{n} \sum_{x \in \ZZ} \nabla_n f  \big( \tfrac{x}{n} \big) e^{2i \pi y x}.
\end{equation}
Then we iterate this $p$ times. To prove (\ref{eq:dibp}), we perform a discrete integration by parts. Let us define for any $x \in \ZZ$
\begin{equation*}
D_x =  \frac{e^{i \pi yx} \sin (\pi (x+1) y)}{\sin(\pi y)}, \quad {\widetilde D}_x =  \frac{e^{i \pi yx}  \sin (\pi (1-x) y)}{\sin(\pi y)}.
\end{equation*}
Observe that for any $x \in \ZZ$, $D_x +{\widetilde D}_{x+1} =1$ and that
\begin{equation*}
\begin{split}
D_x =\sum_{k=0}^x e^{2i \pi k y}, \quad x \ge 0,\\
{\widetilde D}_x= \sum_{k=x}^0 e^{2i \pi k x}, \quad x\le 0.
\end{split}
\end{equation*}
Then we write
\begin{equation*}
\begin{split}
\frac{1}{n} \sum_{x \in \ZZ} f\big( \tfrac{x}{n} \big) e^{2i \pi y x} &=\frac{1}{n} \sum_{x \ge 1} f\big( \tfrac{x}{n} \big) (D_x -D_{x-1}) +  \frac{1}{n} \sum_{x \le -1} f\big( \tfrac{x}{n} \big) ({\widetilde D}_x -{\widetilde D}_{x+1}) +\frac{f (0) }{n}\\
&= -\frac{1}{n^2} \sum_{x \ge 0} \nabla_n f\big( \tfrac{x}{n} \big) D_x + \frac{1}{n^2} \sum_{x \le -1} \nabla_n f\big( \tfrac{x-1}{n} \big) {\widetilde D}_{x} -\frac{f(-1/n)}{n}\\
&=  -\frac{1}{n^2} \sum_{x \in \ZZ} \nabla_n f\big( \tfrac{x}{n} \big) D_x + \frac{1}{n^2} \sum_{x \le -1} \nabla_n f\big( \tfrac{x}{n} \big) ({\widetilde D}_{x+1}+D_x) -\tfrac{f(0)}{n}\\
&= -\frac{1}{n^2} \sum_{x \in \ZZ} \nabla_n f\big( \tfrac{x}{n} \big) D_x + \frac{1}{n^2} \sum_{x \le -1} \nabla_n f\big( \tfrac{x}{n} \big)  -\frac{f(0)}{n}\\
&=-\frac{1}{n^2} \sum_{x \in \ZZ} \nabla_n f\big( \tfrac{x}{n} \big) D_x,
\end{split}
\end{equation*}
where the last equality is due to the fact that we have a telescopic sum. Using the explicit expression of $D_x$ and again a telescopic argument we get (\ref{eq:dibp}).

Now, for the second claim, we observe that if $f \in {\mc S} (\RR)$, the assumption (\ref{eq:asssfp}) is satisfied. Moreover, for any $|y| \le 1/2$, $|{\widehat f}_n (yn)| \le C$ for a constant $C$ independent of $n$ and $y$. By using the first claim proved above we deduce that there exists a constant $C:=C(p,f)$ such that
\begin{equation*}
| \widehat{f_n} (yn) | \le C \inf\left\{ 1, \frac{1}{n^p |\sin(\pi y)|^p} \right\}\le\frac{C'}{1+ (n|y|)^p}.
\end{equation*}
We notice that from the previous estimate we also get that
\begin{equation*}
| \widehat{f_n} (yn) |^2 \le \frac{C'}{1+ (n|y|)^p},
\end{equation*}
which will be useful in what follows.

\end{proof}

\section{Some computations involving trigonometric polynomials}
\label{sec:Ab}

The Fourier transform of the function $\Delta_n h$ for a given, summable function $h: \sfrac{1}{n} \bb Z^2 \to \bb R$ is given by:
\begin{equation}
\begin{split}
\widehat{(\Delta_n h)}_n (k,\ell)
		&= \frac{1}{n^2}  \sum_{\mathclap{x,y \in \bb Z}}  \Delta_n h \big(\tfrac{\vphantom{y}x}{n},\tfrac{y}{n}\big) e^{ \tfrac{2\pi i  (kx+\ell y)}{n}}\\
		&= n^2\big( e^{\tfrac{2\pi i k}{n}}\!\! + e^{-\tfrac{2\pi i k}{n}}\!\! + e^{\tfrac{2 \pi i \ell}{n}} \!\!+ e^{-\tfrac{2\pi i \ell}{n}} \!\!-4\big) \widehat{h}_n(k,\ell)\\
		&= - n^2\Lambda \big( \tfrac{k}{n}, \tfrac{\ell}{n}\big) \widehat{h}_n(k,\ell),
\end{split}
\end{equation}
where
\vspace{-10pt}
\begin{equation}
\begin{split}
 \Lambda \big( \tfrac{k}{n}, \tfrac{\ell}{n}\big)&= -\big( e^{\frac{2\pi i k}{n}}\!\! + e^{-\frac{2\pi i k}{n}}\!\! + e^{\frac{2 \pi i \ell}{n}} \!\!+ e^{-\frac{2\pi i \ell}{n}} \!\!-4\big)\\
		&=4 \left[ \sin^2\big(\tfrac{\pi k}{n}\big) + \sin^2\big(\tfrac{\pi \ell}{n}\big)\right].
\end{split}
\end{equation}
Similarly, the Fourier transform of ${\mc A}_n h$ is given by

\begin{equation}
\begin{split}
\widehat{({\mc A}_n h)}_n (k,\ell)=i\,n \,\Omega \big( \tfrac{k}{n}, \tfrac{\ell}{n}\big) \widehat{h}_n(k,\ell),
\end{split}
\end{equation}
 where
\begin{equation}
\begin{split}
i \,\Omega \big( \tfrac{k}{n}, \tfrac{\ell}{n}\big)&= e^{\tfrac{2\pi i k}{n}}\!\! + e^{\tfrac{2\pi i \ell}{n}}\!\! -e^{-\tfrac{2\pi i k}{n}}\!\! - e^{-\tfrac{2 \pi i \ell}{n}}\\
		&= 2\,i\,\big( \sin\big(\tfrac{2\pi k}{n}\big) + \sin \big(\tfrac{2 \pi \ell}{n} \big)\big).
\end{split}
\end{equation}
%
Notice in particular that $\Omega(\frac{k}{n}, \frac{\ell}{n})$ is a real number.
Let us now compute the Fourier transform of the function $g_n= \nabla_n f \otimes \delta$ defined in (\ref{eq:3.9}):
\begin{equation}
\begin{split}
\widehat{g}_n(k,\ell)
		& = \cfrac{1}{n^2} \sum_{x,y\in \mathbb{Z}} \big[ \nabla_n f  \otimes \delta \big] \big( \tfrac{\vphantom{y}x}{n}, \tfrac{y}{n}\big) e^{ \tfrac{2i \pi(kx + \ell y)}{n}}\\
		&= - \cfrac{i n}{2} \Omega \big( \tfrac{k}{n}, \tfrac{\ell}{n}\big) {\widehat{f_n}} (k+ \ell).
\end{split}
\end{equation}
Several times we will use the following elementary change of variable property.

\begin{lemma}
\label{lem:cov}
Let $F: \RR^2 \to \CC$ be a $n$-periodic function in each direction of $\RR^2$. Then we have that
\begin{equation*}
\iint_{[-\tfrac{n}{2}, \tfrac{n}{2}]^2} F (k,\ell) \, dk d\ell \; = \; \iint_{[-\tfrac{n}{2}, \tfrac{n}{2}]^2} F(\xi- \ell, \ell) \, d\xi d\ell.
\end{equation*}
\end{lemma}

\begin{proof}
Let us write $\chi (x,y) = {\bf 1}_{\{x,  y-x \in [-1/2 , 1/2]\}}$. We have that
\begin{equation*}
\begin{split}
&\iint_{[-n/2, n/2]^2} F (k,\ell) \; dk d\ell =\int_{-n}^n \;  \left\{ \int F(u-\ell,\ell)  \chi \big (\tfrac{\ell}{n}, \tfrac{u}{n} \big) d\ell \right\} du \\
&= \int_{-n}^0 \;  \left\{ \int_{-n/2}^{u+n/2}  F (u-\ell,\ell) d\ell \right\} du + \int_0^n \; \left\{ \int_{u-n/2}^{n/2}  F (u-\ell,\ell) d\ell \right\} du \\
&= \int_{0}^n \; \left\{ \int_{-n/2}^{u-n/2}  F (u-\ell,\ell) d\ell \right\} du + \int_0^n \;  \left\{ \int_{u-n/2}^{n/2}  F (u-\ell,\ell)  d\ell \right\} du \\
&=  \int_{0}^n \;  \left\{ \int_{-n/2}^{n/2}  F (u-\ell,\ell) d\ell \right\} du\\
&=  \int_{-n/2}^{n/2} \;  \left\{ \int_{-n/2}^{n/2}  F (u-\ell,\ell) d\ell \right\} du.
\end{split}
\end{equation*}
\end{proof}

\section{Proof of Lemma \ref{l3}}
\label{sec:proof_l3}
Let $h_n: \sfrac{1}{n} \bb Z^2 \to \bb R$ be the unique solution in $\ell^2(\sfrac{1}{n}\bb Z^2)$ of \eqref{ec4.22}. Observe that $h_n$ is a symmetric function. The Fourier transform of $h_n$ is not difficult to compute by using Appendix \ref{sec:Ab}. In fact, we have that
\begin{equation}
\label{ecB.4}
\widehat{h}_n(k,\ell) =  \frac{1}{2\sqrt n} \frac{i\,\Omega \big( \tfrac{k}{n}, \tfrac{\ell}{n}\big)}{\Lambda\big( \tfrac{k}{n}, \tfrac{\ell}{n}\big) - i\, \Omega \big( \tfrac{k}{n}, \tfrac{\ell}{n}\big)} \; \widehat{f_n}(k+\ell).
\end{equation}
Our aim will be to study the behavior of $h_n$, as $n \to \infty$, and in particular to prove Lemma \ref{l3}.

\subsection{Proof of (\ref{eq:l31})}
Observe first that
\begin{equation}
i\,\Omega \big( \tfrac{\xi -\ell}{n}, \tfrac{\ell}{n}\big) = e^{ \tfrac{2i \pi\ell}{n}} (1- e^{-\tfrac{2 i \pi \xi}{n}}) -  e^{- \tfrac{2i \pi \ell}{n}} (1- e^{\tfrac{2 i \pi \xi}{n}})
\end{equation}
so that
\begin{equation}
\label{eq:omegaomega}
\Omega \big( \tfrac{\xi -\ell}{n}, \tfrac{\ell}{n}\big)^2 \le 4 \Big| 1- e^{\tfrac{2 i \pi \xi}{n}} \Big|^2=16 \sin^2 \big(\tfrac{\pi\xi}{n}\big).
\end{equation}
Then, by Plancherel-Parseval's relation and by using Lemma \ref{lem:cov} we have that
\begin{equation*}
\begin{split}
\| h_n \|^2_{2,n}
		&=\iint_{[-\tfrac{n}{2}, \tfrac{n}{2}]^2} | {\widehat h}_n (k, \ell) |^2 dk d\ell \\
		&= \frac{1}{4n} \iint_{[-\tfrac{n}{2}, \tfrac{n}{2}]^2}\frac{\Omega\big( \tfrac{k}{n}, \tfrac{\ell}{n}\big)^2 \, |{\widehat f_n} (k+\ell)|^2}{\Lambda \big( \tfrac{k}{n}, \tfrac{\ell}{n}\big)^2 + \Omega \big( \tfrac{k}{n}, \tfrac{\ell}{n}\big)^2} \; dk d\ell \\
		&\le \frac{1}{n} \int_{-n/2}^{n/2} \big|1 - e^{\tfrac{2i\pi\xi}{n}}\big|^2 \big|{\widehat f}_n (\xi)\big|^2 \; \left[ \int_{-n/2}^{n/2} \frac{d\ell}{\Lambda \big( \tfrac{\xi - \ell}{n}, \tfrac{\vphantom{\xi}\ell}{n}\big )^2 + \Omega \big( \tfrac{\xi - \ell}{n}, \tfrac{\vphantom{\xi}\ell}{n}\big)^2} \right] \; d\xi \\
		&= 4 n \int_{-1/2}^{1/2} \sin^2 (\pi y) |{\widehat f}_n (ny) |^2 W(y) dy,
\end{split}
\end{equation*}
where for the last equality we performed the changes of variables $y=\frac{\xi}{n}$ and $x= \frac{\ell}{n}$. The function $W$ is defined by
\begin{equation}
\label{eq:W}
W(y) =  \int_{-1/2}^{1/2}  \frac{dx}{\Lambda (y-x,x)^2 + \Omega (y-x, x)^2 }.
\end{equation}
Since by Lemma \ref{lem:W} we have that $W(y) \le C |y|^{-3/2}$ on $[-\tfrac{1}{2}, \tfrac{1}{2}]$, we get, by using the second part of Lemma \ref{lem:sfp} with $p=3$ and the elementary inequality $\sin^2 (\pi y) \le (\pi y)^2$, that
\begin{equation*}
\begin{split}
&\iint_{[-\tfrac{n}{2}, \tfrac{n}{2}]^2} | {\widehat h}_n (k, \ell) |^2 dk d\ell \le C' n \int_{-1/2}^{1/2} \cfrac{|y|^{1/2}}{1+ (n|y|)^3} dy = O (n^{-1/2}).
\end{split}
\end{equation*}

\subsection{Proof of (\ref{eq:l32})}
\label{subsec:firstterm}

We denote by $G$ the $1$-periodic function defined by
\begin{equation}
\label{eq:fF}
G(y)= \frac{1}{4} \int_{-1/2}^{1/2} \cfrac{\Omega(y-z,z)^2 }{\Lambda (y-z,z) -i\, \Omega (y-z,z)}\, dz.
\end{equation}
As $y\to 0$, the function $G$ is equivalent (in a sense defined below) to the function $G_0$ given by
\begin{equation}
\label{eq:G0}
G_0 (y)= \cfrac{|\pi y|^{3/2}}{2} ( 1+i \, {\rm{sgn}} (y)) .
\end{equation}
In fact, we show in Lemma \ref{lem:G0345} that there exists a constant $C>0$ such that for any $|y| \le 1/2$
\begin{equation}
\label{eq:G000}
|G (y) -G_0 (y)| \le C |y|^2.
\end{equation}
We denote by $\mc F f$ the (continuous) Fourier transform of $f$, defined by
\begin{equation}
{\mc F} f (y) = \int_{-\infty}^{+\infty} \, f(t)\;  e^{2i\pi t y }\;  dt
\end{equation}
and by $q:=q(f): \RR \to \RR$ the function defined by
\begin{equation}
q(x) = \int_{- \infty}^{\infty} e^{-2 i \pi xy} G_0 (y) {\mc F} f (y) dy
\end{equation}
which coincides with $-\tfrac{1}{4} {\bb L} f (x)$.

Let $q_n : {\sfrac{1}{n}} \ZZ \to \RR$ the function defined by
\begin{equation}
q_n \big(\tfrac{x}{n}) = {{\mc D}}_n h_n \,  \big(\tfrac{x}{n}).
\end{equation}

\begin{lemma}
We have
\begin{equation}
\lim_{n \to + \infty} \frac{1}{n} \sum_{x \in \ZZ} \left[ q \big(\tfrac{x}{n} \big) -q_n \big( \tfrac{x}{n}\big)\right]^2 =0.
\end{equation}
\end{lemma}

\begin{proof}
Since ${\widehat h}_n$ is a symmetric function we have
\begin{equation*}
\begin{split}
{\widehat{q}_n} (\xi)
		&= \sum_{x \in \ZZ} e^{ \tfrac{2i \pi \xi x}{n}} \big[ h_n \big( \tfrac{x}{n}, \tfrac{x+1}{n} \big) - h_n \big( \tfrac{x-1}{n}, \tfrac{x}{n} \big) \big]\\
		& =\sum_{x \in \ZZ} e^{ \tfrac{2i \pi \xi x}{n}} \iint_{[-\tfrac{n}{2}, \tfrac{n}{2}]^2} e^{-\tfrac{2i \pi (k+\ell)x}{n}} \Big[ e^{- \tfrac{2i \pi\ell}{n} } - e^{ \tfrac{2i \pi \ell}{n}} \Big] {\widehat h_n} (k, \ell) \, dk d\ell\\
		&=\frac{1}{2}\sum_{x \in \ZZ} e^{ \tfrac{2i \pi \xi x}{n}} \iint_{[-\tfrac{n}{2}, \tfrac{n}{2}]^2} e^{-\tfrac{2i \pi (k+\ell) x}{n}} \Big[ e^{- \tfrac{2i \pi \ell}{n} } - e^{ \tfrac{2i \pi \ell}{n}} \Big] {\widehat h_n} (k, \ell) \, dk d\ell \\
		&\quad +\frac{1}{2}\sum_{x \in \ZZ} e^{ \tfrac{2i \pi \xi x}{n}} \iint_{[-\tfrac{n}{2}, \tfrac{n}{2}]^2} e^{-\tfrac{2i \pi (k+\ell) x}{n}} \Big[ e^{-\tfrac{2i \pi k}{n} } - e^{\tfrac{2i \pi k}{n}} \Big] {\widehat h_n} (k, \ell) \, dk d\ell \\
		&= -\frac{i}{2}  \sum_{x \in \ZZ} e^{ \tfrac{2i \pi \xi x}{n}} \iint_{[-\tfrac{n}{2}, \tfrac{n}{2}]^2} e^{-\tfrac{2i \pi (k+\ell) x}{n}}  \Omega \big( \tfrac{k}{n}, \tfrac{\ell}{n}\big)  {\widehat h_n} (k,\ell) \, dk d\ell.
\end{split}
\end{equation*}
We use now Lemma \ref{lem:cov} and the inverse Fourier transform relation to get
\begin{equation*}
\begin{split}
{\widehat{q}_n} (\xi) =- \cfrac{in}{2} \int_{-n/2}^{n/2} \Omega  \big( \tfrac{\xi -\ell}{n}, \tfrac{\ell}{n}\big)  {\widehat h_n} (\xi- \ell,\ell) \, d\ell.
\end{split}
\end{equation*}
By the explicit expression (\ref{ecB.4}) of ${\widehat h_n}$ we obtain that
\begin{equation*}
{\widehat{q_n}} (\xi) = \frac{\sqrt{n}}{4} \left[\int_{-n/2}^{n/2} \frac{\Omega  \big( \tfrac{\xi -\ell}{n}, \tfrac{\ell}{n}\big)^2}{\Lambda \big( \tfrac{\xi - \ell}{n}, \tfrac{\ell}{n}\big )-i\Omega \big( \tfrac{\xi - \ell}{n}, \tfrac{\ell}{n}\big)} \, d\ell \right]\; {\widehat f_n} (\xi).
\end{equation*}
Again by the inverse Fourier transform we get that
\begin{equation*}
q_n\big(\tfrac{x}{n}\big)=\int_{-n /2}^{n/2}\; e^{- \tfrac{2i\pi \xi x}{n}}  n^{3/2} G \big( \tfrac{\xi}{n}\big)  {\widehat f}_n (\xi)   \; d\xi.
\end{equation*}
Then we have
\begin{equation}
\begin{split}
q \big(\tfrac{x}{n} \big) -q_n \big( \tfrac{x}{n}\big)
		&= \int_{|\xi| \ge n/2} \; e^{- \tfrac{2i\pi \xi x}{n}}  \; G_0 (\xi)  \;  \mc F f (\xi)   \; d\xi\\
		&\quad + \int_{|\xi| \le n/2} \; e^{- \tfrac{2i\pi \xi x}{n}}  \; G_0 (\xi)  \; \big[  \mc F f (\xi)  - {\widehat f}_n (\xi)  \big]\; d\xi\\
		&\quad \quad +n^{3/2}\int_{|\xi| \le n/2} \; e^{- \tfrac{2i\pi \xi x}{n}}  (G_0- G) \big( \tfrac{\xi}{n}\big) {\widehat f}_n (\xi) \; d\xi.
\end{split}
\end{equation}
Above we have used the fact that $n^{3/2} G_0\big(\tfrac{\xi}{n})=G_0(\xi)$.
Then we use the triangular inequality and Plancherel's theorem in the two last terms of the right hand side to get
\begin{equation}
\label{eq:l2limit}
\begin{split}
 \frac{1}{n} \sum_{x \in \ZZ} \big[ q \big(\tfrac{x}{n} \big) -q_n \big( \tfrac{x}{n}\big)\big]^2
 		&\le \frac{1}{n} \sum_{x \in \ZZ} \left| \int_{|\xi| \ge n/2} \; e^{- \tfrac{2i\pi \xi x}{n}}  \; G_0 (\xi)  \;  (\mc F f) (\xi)   \; d\xi \right|^2\\
		&\quad + \int_{|\xi| \le n/2} \left|  G_0 (\xi)  \; \big[  \mc F f (\xi)  - {\widehat f}_n (\xi)  \big]\; \right|^2 d\xi\\
		&\quad \quad +n^3 \int_{|\xi| \le n/2} \; \left|  (G_0- G) \big( \tfrac{\xi}{n}\big) {\widehat f}_n (\xi) \right|^2 \; d\xi \\
		&= (I) +(II)+(III).
\end{split}
\end{equation}
The contribution of the term $(I)$ is estimated by performing an integration by parts and using the fact that the Fourier transform ${\mc F} f$ of $f$ is in the Schwartz space and that $G_0$ and $G_0'$ grow at most polynomially:
\begin{equation}
\begin{split}
(I)&\leq \frac{C}{n} \sum_{x \in \ZZ}\frac{n^2}{|x|^2} \left\{   \big|  (G_0 \, \mc F f )(\pm \tfrac{n}{2}) \big|^2  + \left| \int_{|\xi| \ge n/2} \big| \tfrac{d}{d\xi} [G_0 \,  \mc F f] (\xi) \big| \; d\xi \right|^2 \right\}.
\end{split}
\end{equation}

Then one can get that $(I) \le  C_p  n^{-p}$ for any $p \ge 1$ with a suitable constant $C_p>0$. Therefore $(I)$ gives a trivial contribution in (\ref{eq:l2limit}).  The term $(II)$ in (\ref{eq:l2limit}) can be bounded above by a constant times
\begin{equation}
\int_{-n/2}^{n/2} |\xi|^3 |\mc F f (\xi) -{\widehat f}_n (\xi) \big|^2 d\xi
\end{equation}
because $|G_0 (\xi)| \le C |\xi|^{3/2}$ for any $\xi$. Let $0<A<n/2$ and write
\begin{equation}
\label{eq:zurich}
\begin{split}
\int_{-n/2}^{n/2} |\xi|^3  |\mc F f (\xi) -{\widehat f}_n (\xi) \big|^2 d\xi
		&= \int_{|\xi| \le A} |\xi|^3  \big|\mc F f (\xi) -{\widehat f}_n (\xi)\big|^2 d\xi \\
		&\quad+ \int_{A\le |\xi| \le n/2} |\xi|^3 \big|\mc F f (\xi) -{\widehat f}_n (\xi) \big|^2 d\xi.
\end{split}
\end{equation}
Now, performing a change of variables $\xi=\frac{y}{n}$ and using the fact that $f$ is in the Schwartz space and Lemma \ref{lem:sfp}, the second term on the right hand side of (\ref{eq:zurich}) is bounded above by
\begin{multline*}
C \int_{|\xi| \ge A} |\xi|^3 \, \big|{\mc F} f(\xi)\big|^2 d\xi + Cn^4  \int_{\tfrac{A}{n} \le |y| \le 1/2} \frac{|y|^3}{1+|ny|^p} dy \\
\le C \int_{|\xi| \ge A} |\xi|^3 \, \big|{\mc F} f(\xi)\big|^2 d\xi + C \int_{A}^{\infty} \frac{z^3}{1+z^p} dz:=\ve (A),
\end{multline*}
where $p$ is bigger than $4$ and $C$ is independent of $n$ and $A$. Observe that $\ve (A) \to0$, as $A\to \infty$. It follows that the left hand side of (\ref{eq:zurich}) is bounded above by
\begin{equation*}
\int_{|\xi| \le A} |\xi|^3 \big|\mc F f (\xi) -{\widehat f}_n (\xi)\big|^2 d\xi +\ve (A).
\end{equation*}
We first take the limit $n \to  \infty$ and use the dominated convergence theorem for the first term of the expression above and then we take the limit as $A \to  \infty$.

The contribution of $(III)$ is estimated by using (\ref{eq:G000}) which gives
\begin{equation*}
(III) \le \frac{C}{n} \int_{|\xi| \le n/2} \; | \xi|^4|  \; | {\widehat f}_n (\xi) |^2 \; d\xi = Cn^4 \int_{-1/2}^{1/2} |z|^4 | {\widehat{f}_n} (nz) |^2 dz
\end{equation*}
which goes to $0$, as $n\to\infty$, by Lemma \ref{lem:sfp} applied with $p=2$.
\end{proof}

\section{Proof of Lemma \ref{lem:5,7.}}
\label{subsec:secondterm}
Let $w_n$ be defined by
\begin{equation}
w_n \big(\tfrac{x}{n}\big) = h_n \big( \tfrac{x}{n} , \tfrac{x+1}{n} \big) -h_n \big( \tfrac{x}{n} , \tfrac{x}{n} \big)
\end{equation}
and observe that
\begin{equation*}
\tfrac{1}{\sqrt{n}} {\widetilde{\mc D}_n} h_n \; \big( \tfrac{\vphantom{y}x}{n}, \tfrac{y}{n} \big) = n^{3/2}
\begin{cases}
w_n \big(\tfrac{x}{n} \big), \quad y =x+1,\\
w_n \big(\tfrac{x-1}{n} \big), \quad y =x-1,\\
0, \quad \text{otherwise}.
\end{cases}
\end{equation*}
Now, since
\begin{equation*}
\begin{split}
\widehat{\widetilde{\mc D}_n{h}_n}(k,l) =& \sum_{x\in\mathbb{Z}}\Big\{w_n \big(\tfrac{x}{n}\big)e^{\tfrac{2\pi i (kx + \ell (x+1))}{n}}+w_n \big(\tfrac{x}{n}\big)e^{\tfrac{2 i \pi (k(x+1)+\ell x)}{n}}\Big\}\\
=&n \Big\{e^{\tfrac{2i\pi\ell}{n}}+e^{\tfrac{2i\pi k}{n}}\Big\}\widehat{w}_n(k+\ell)
\end{split}
\end{equation*}
and using the computations of the Appendix \ref{sec:Ab}, it is easy to see that the Fourier transform ${\widehat v}_n$ is given by
\begin{equation}\label{FT of vn}
{\widehat v}_n (k, \ell) = -\cfrac{1}{n} \, \cfrac{e^{ \tfrac{2i \pi k}{n}} +e^{\tfrac{2i \pi \ell}{n}} }{  \Lambda\big( \tfrac{k}{n}, \tfrac{\ell}{n}\big) -i\, \Omega\big( \tfrac{k}{n}, \tfrac{\ell}{n}\big) } \, {\widehat w}_n (k+\ell).
\end{equation}
By using Lemma \ref{lem:cov}, we have that the Fourier transform of $w_n$ is given by
\begin{equation}
\begin{split}
{\widehat w}_n (\xi) &= \frac{1}{n} \sum_{x\in\mathbb{Z}} e^{ \tfrac{2i\pi \xi x}{n}} \iint_{[-\tfrac{n}{2}, \tfrac{n}{2}]^2} {\widehat h}_n (k,\ell) e^{- \tfrac{2 i \pi (k+\ell) x}{n}} \big\{ e^{- \tfrac{2i\pi\ell}{n}} -1 \big\} \; dk d\ell \\
&= \frac{1}{n} \sum_{x\in\mathbb{Z}} e^{ \tfrac{2i\pi \xi x}{n}} \int_{-n/2}^{n/2} \; e^{- \tfrac{2 i \pi u x}{n}} \left\{ \int_{-n/2}^{n/2}  {\widehat h}_n (u-\ell,\ell)\big\{ e^{- \tfrac{2i\pi \ell}{n}} -1 \big\} d\ell \right\} du\\
&=  \int_{-n/2}^{n/2}  {\widehat h}_n (\xi-\ell,\ell)\big\{ e^{- \tfrac{2i\pi \ell}{n}} -1 \big\} \, d\ell.
\end{split}
\end{equation}
In the last line we used the inverse Fourier transform.
By (\ref{ecB.4}) we get
\begin{equation}\label{Est}
\begin{split}
{\widehat w}_n (\xi) &= -\frac{1}{2\sqrt{n}} \, {\widehat f}_n (\xi) \,  \int_{-n/2}^{n/2}  \frac{\big(1- e^{- \tfrac{2i\pi \ell}{n}} \big)\,  i\, \Omega \big( \tfrac{\xi - \ell}{n}, \tfrac{\vphantom{\xi}\ell}{n}\big)} { \Lambda\big( \tfrac{\xi - \ell}{n}, \tfrac{\vphantom{\xi}\ell}{n}\big) - i \, \Omega\big( \tfrac{\xi - \ell}{n}, \tfrac{\vphantom{\xi}\ell}{n}\big)}\, d\ell\\
&= -\frac{\sqrt{n}}{2} I\big( \tfrac{\xi}{n} \big) {\widehat f}_n (\xi)
\end{split}
\end{equation}
where the function $I$ is defined by
\begin{equation}
\label{eq:I}
I(y)= \int_{-1/2}^{1/2} \cfrac{(1-e^{-2i \pi x}) \, i \, \Omega (y-x,x)}{\Lambda(y-x,x) -i \Omega(y-x,x)} dx.
\end{equation}

\subsection{Proof of \eqref{est vn}}
By Plancherel-Parseval's relation and Lemma \ref{lem:cov} we have
\begin{equation*}
\begin{split}
\|v_n\|_{2,n}^2
		&= \iint_{[-\tfrac{n}{2},\tfrac{n}{2}]^2} |{\widehat v_n} (k, \ell)|^2 dk d\ell\\
		&= \frac{1}{n^2} \int_{-n/2}^{n/2} |\widehat w_n (\xi)|^2 \,  \int_{-n/2}^{n/2} \left|\cfrac{e^{ \tfrac{2i \pi(\xi - \ell)}{n} } +e^{ \tfrac{2 i \pi \ell}{n} }} {\Lambda\big(\tfrac{\xi -\ell}{n}, \tfrac{\vphantom{\xi}\ell}{n} \big) -i \Omega\big(\tfrac{\xi -\ell}{n}, \tfrac{\vphantom{\xi}\ell}{n} \big) }\,  \right|^2 d\ell d\xi\\
		&\le \frac{C}{n} \int_{-n/2}^{n/2} |\widehat w_n (\xi)|^2 W \big( \tfrac{\xi}{n}\big) d\xi\\
		&=\cfrac{C }{4}  \int_{-n/2}^{n/2} \big|\widehat f_n (\xi)\big|^2 \big| \, I \big( \tfrac{\xi}{n}\big) \big|^2   W \big( \tfrac{\xi}{n}\big) d\xi\\
		&=\cfrac{C n}{4}  \int_{-1/2}^{1/2} |\widehat f_n (ny)|^2 | \, I \big(y\big) |^2   W \big( y\big) d\xi,
\end{split}
\end{equation*}
where in the third inequality we used the Cauchy-Schwarz inequality, in the penultimate inequality we used \eqref{Est} and in the last equality we used a change of variables. Recall that the function $W$ is defined by (\ref{eq:W}). By Lemma \ref{lem:W}, Lemma \ref{lem:I} and Lemma \ref{lem:sfp}, we get, that
\begin{equation*}
\begin{split}
\|v_n\|_{2,n}^2 & \le Cn \int_{-1/2}^{1/2} |{\widehat f}_n (ny)|^2 |\sin (\pi y)|^{3/2} dy\\
& \le  C \int_{-1/2}^{1/2}  \cfrac{|y|^{3/2}}{1+|ny|^p} dy = \cfrac{C}{n^{3/2}} \int_{-n/2}^{n/2} \cfrac{|z|^{3/2}}{1+|z|^p} dz,
\end{split}
\end{equation*}
which goes to $0$ as soon as $p$ is chosen bigger than $3$.

\subsection{Proof of \eqref{est der vn}}

Notice that
\begin{equation}
\begin{split}
{\widehat{{\mc D}_n v_n}}(\xi)
		&= \sum_{x\in\mathbb{Z}}\big\{v_n\big(\tfrac{x}{n},\tfrac{x+1}{n}\big)-v_n\big(\tfrac{x-1}{n},\tfrac{x}{n}\big)\big\}e^{ \tfrac{2i \pi \xi}{n}}\\
		&=\sum_{x\in\mathbb{Z}}v_n\big(\tfrac{x}{n},\tfrac{x+1}{n}\big)e^{\tfrac{2i\pi \xi x}{n}}\big(1-e^{\tfrac{2i\pi \xi }{n}}\big)\\
		&=\sum_{x\in\mathbb{Z}}e^{\tfrac{2i\pi \xi x}{n}}\big(1-e^{\tfrac{2i\pi\xi }{n}}\big)\iint_{[-\tfrac{n}{2},\tfrac{n}{2}]^2} {\widehat v_n} (k, \ell) e^{-\tfrac{2 i \pi(kx+\ell(x+1))}{n}}dk d\ell\\
		&=\sum_{x\in\mathbb{Z}}e^{\tfrac{2i\pi \xi x}{n}}\big(1-e^{\tfrac{2i\pi \xi }{n}}\big)\iint_{[-\tfrac{n}{2},\tfrac{n}{2}]^2} {\widehat v_n} (k, \ell) e^{-\tfrac{2i\pi(kx+\ell(x+1))}{n}}dk d\ell.\\
\end{split}
\end{equation}
Now, by Lemma \ref{lem:cov} we get
\begin{equation}
\begin{split}
{\widehat{{\mc D}_n v_n}}(\xi)
		&=\sum_{x\in\mathbb{Z}}e^{\tfrac{2i\pi \xi x}{n}}\big(1-e^{\tfrac{2i\pi \xi }{n}}\big)\iint_{[-\tfrac{n}{2},\tfrac{n}{2}]^2} {\widehat v_n} (m-\ell, \ell)e^{-\tfrac{2i\pi\ell}{n}} e^{-\tfrac{2i\pi mx}{n}}dm d\ell\\
		&=n\big(1-e^{\tfrac{2i\pi \xi }{n}}\big)\int_{-\tfrac{n}{2}}^{\tfrac{n}{2}} {\widehat v_n} (\xi-\ell, \ell)e^{-\tfrac{2i\pi\ell}{n}} d\ell\\
		&=-\big(1-e^{\tfrac{2i\pi \xi }{n}}\big){\widehat w_n} (\xi)\int_{-\tfrac{n}{2}}^{\tfrac{n}{2}}\frac{1+e^{ \tfrac{2i \pi(\xi-2\ell)}{n}}}{ \Lambda \big( \tfrac{\xi-\ell}{n}, \tfrac{\vphantom{\xi}\ell}{n}\big) +i \Omega \big( \tfrac{\xi-\ell}{n}, \tfrac{\vphantom{\xi}\ell}{n}\big)}d\ell\\
		&=-{n}\big(1-e^{\tfrac{2i\pi \xi }{n}}\big){\widehat w_n} (\xi)J\big (\tfrac{\xi}{n}\big),
\end{split}
\end{equation}
where in the penultimate equality we used \eqref{FT of vn} and in last equality we performed a change of variables. Above, $J$ is given by
\begin{equation}
\label{eq:J}
J(y) = \int_{-1/2}^{1/2} \cfrac{1+e^{2i \pi (y-2x)}}{\Lambda(y-x,x)-i\Omega(y-x,x) } dx.
\end{equation}
Now, by using \eqref{Est} we get, finally, that
\begin{equation}
{\widehat{{\mc D}_n v_n}}(\xi)=\frac{n^{3/2}}{2}\big( 1- e^{ \tfrac{2i \pi \xi}{n}} \big)  {\widehat f}_n (\xi) I \big (\tfrac{\xi}{n}\big) J\big (\tfrac{\xi}{n}\big),
\end{equation}
where $I$ is defined by (\ref{eq:I}).

By Plancherel-Parseval's relation we have to prove that
\begin{multline}
n^3 \int_{-n/2}^{n/2} \sin^2 \big (\pi \tfrac{\xi}{n}\big )\big| {\hat f}_n (\xi) \big|^2 \big| I \big (\tfrac{\xi}{n}\big) \big|^2 \big| J\big (\tfrac{\xi}{n}\big) \big|^2    d\xi=\\
= n^4 \int_{-1/2}^{1/2} \sin^2 (\pi y) | I(y)|^2 | J(y)|^2 |\hat f_n (ny)|^2 dy
\end{multline}
vanishes, as $n\to\infty$. By Lemma \ref{lem:sfp}, Lemma \ref{lem:I} and Lemma \ref{lem:J}, this is equivalent to show that the following term goes to 0, as $n\to\infty$:
\begin{equation*}
n^4 \int_{-1/2}^{1/2} \cfrac{|y|^4}{1+|ny|^p} dy = \cfrac{1}{n} \int_{-n/2}^{n/2} \cfrac{|z|^4}{1+|z|^p} dz.
\end{equation*}
But for $p$ bigger than $5$, this term goes to $0$, as $n\to\infty$.

\subsection{Proof of \eqref{est tilde d vn}}
Let $\theta_n:\sfrac{1}{n}\ZZ \to \RR$ be defined by
\begin{equation*}
\theta_n \big( \tfrac{x}{n} \big) = v_{n} \big(\tfrac{x}{n}, \tfrac{x+1}{n} \big) - v_n \big( \tfrac{x}{n},\tfrac{x}{n} \big)
\end{equation*}
and observe that
\begin{equation*}
\tfrac{1}{\sqrt{n}} {\widetilde{\mc D}_n} v_n \; \big( \tfrac{\vphantom{y}x}{n}, \tfrac{y}{n} \big) = n^{3/2}
\begin{cases}
\theta_n \big(\tfrac{x}{n} \big), \quad y =x+1,\\
\theta_n \big(\tfrac{x-1}{n} \big), \quad y =x-1,\\
0, \quad \text{otherwise}.
\end{cases}
\end{equation*}
Now, doing similar computations as above we have that
\begin{equation}
\begin{split}
{\widehat \theta_n} (\xi)&= \frac{1}{n}\sum_{x\in\mathbb{Z}}\big\{v_n\big(\tfrac{x}{n},\tfrac{x+1}{n}\big)-v_n\big(\tfrac{x}{n},\tfrac{x}{n}\big)\big\}e^{ \tfrac{2i \pi\xi}{n}}\\
&=\frac{1}{n}\sum_{x\in\mathbb{Z}}e^{\tfrac{2i\pi\xi x}{n}}\iint_{[-\tfrac{n}{2},\tfrac{n}{2}]^2} {\widehat v_n} (k, \ell)\big\{e^{-\tfrac{2i\pi(kx+\ell(x+1))}{n}}-e^{\tfrac{2i\pi(k+\ell)x}{n}}\big\}dk d\ell\\
&=\frac{1}{n}\sum_{x\in\mathbb{Z}}e^{\tfrac{2i\pi \xi x}{n}}\iint_{[-\tfrac{n}{2},\tfrac{n}{2}]^2} {\widehat v_n} (m-\ell, \ell)e^{-\tfrac{2i\pi mx}{n}}\{e^{-\tfrac{2i\pi \ell}{n}}-1\}dk d\ell\\
&=\int_{-\tfrac{n}{2}}^{\tfrac{n}{2}} {\widehat v_n} (\xi-\ell, \ell)\{e^{-\tfrac{2i\pi \ell}{n}}-1\} d\ell.\\
\end{split}
\end{equation}
Performing a change of variables and using \eqref{FT of vn} and \eqref{Est} we get that
\begin{equation*}
{\widehat \theta_n} (\xi) = \sqrt{n} {\widehat f}_n (\xi) I \big( \tfrac{\xi}{n} \big) K \big( \tfrac{\xi}{n} \big),
\end{equation*}
where $I$ is defined by (\ref{eq:I}) and $K$ is given by
\begin{equation}
\label{eq:K}
K(y) = \int_{-1/2}^{1/2} \cfrac{(e^{-2i \pi x} -1) (e^{2i \pi (y-x)}+ e^{2i \pi x}  )}{\Lambda(y-x,x) -i \Omega(y-x,x) } dx.
\end{equation}
We need to show that
\vspace{-10pt}
\begin{equation*}
\lim_{n \to \infty} n^2 \| \theta_n \|_{2,n}^2 =0.
\end{equation*}
By Plancherel-Parseval's relation, this is equivalent to prove that
\begin{equation*}
\lim_{n \to \infty} n^3 \int_{-n/2}^{n/2} \big| {\widehat f}_n (\xi) \big|^2 \big|  I \big( \tfrac{\xi}{n} \big) \big|^2  \big| K \big( \tfrac{\xi}{n} \big) \big|^2 \, d\xi =0.
\end{equation*}
By using the change of variables $y=\xi/n$, Lemma \ref{lem:sfp}, Lemma \ref{lem:I} and Lemma \ref{lem:K}, we have
\begin{multline*}
n^3 \int_{-n/2}^{n/2} \big| {\widehat f}_n (\xi) \big|^2 \big|  I \big( \tfrac{\xi}{n} \big) \big|^2  \big| K \big( \tfrac{\xi}{n} \big) \big|^2 \, d\xi \leq \\
 \le Cn^4 \int_{-1/2}^{1/2} \cfrac{|y|^4}{1+|ny|^p} dy  = \cfrac{C}{n} \int_{-n/2}^{n/2} \cfrac{|z|^4}{1+|z|^p} dz
\end{multline*}
which goes to $0$, as $n\to\infty$, for $p$ bigger than $5$.

\section{Asymptotics of few integrals }

\begin{lemma}
\label{lem:G0345}
Recall that $G$ and $G_0$ are defined by (\ref{eq:fF}) and (\ref{eq:G0}). There exists a constant $C>0$ such that for any $|y| \le 1/2$
\begin{equation}
|G (y) -G_0 (y)| \le C |y|^2.
\end{equation}
\end{lemma}

\begin{proof}
We compute the function $G$ by using the residue theorem. For any $y \in [-1/2, 1/2]$ we denote by $w:=w(y)$ the complex number $w=e^{2i\pi y}$. By denoting $z=e^{2i \pi x}$, $x \in [-1/2, 1/2]$, we have that
\begin{equation*}
\begin{split}
& \Lambda (y-x,x)= 4- z (w^{-1} +1) -z^{-1}( w+1),\\
&i \, \Omega (y-x,x)= z(1-w^{-1}) + z^{-1} (w-1).
\end{split}
\end{equation*}
We denote by ${\mc C}$ the unit circle positively oriented. Then, we have
\begin{equation}
G(y)= \frac{1}{16 i \pi} \oint_{\mc C} f_{w} (z) dz
\end{equation}
where the meromorphic function $f_{w}$ is defined by
\begin{equation}
f_w (z)= \frac{[(w-1) +z^2 (1-w^{-1})]^2}{z^2 (z^2 -2 z +w)}.
\end{equation}
The poles of $f_w$ are $0$ and $z_-, z_+$ which are the two solutions of $z^2 -2z +w$. Since
$$1-w=2 |\sin (\pi y)| e^{i \pi [y-\frac{1}{2} {\rm sgn} (y)]}$$
we have that
\begin{equation}
\label{eq:zpm}
z_{\pm} = 1 \pm \sqrt{ 2 |\sin (\pi y)|}\,  e^{ \tfrac{i\pi}{2} [y-\frac{1}{2} {\rm sgn} (y)]}.
\end{equation}
Observe that $|z_-| <1$ and $|z_+|>1$. By the residue theorem, we have
\begin{equation}
\oint_{\mc C} f_{w} (z) dz = 2 \pi i \big[ {\rm {Res}} (f_w, 0) +{\rm {Res}} (f_w, z_-) \big],
\end{equation}
where ${\rm {Res}} (f_w, a)$ denotes the value of the residue of $f_w$ at the pole $a$. An elementary computation shows that
\begin{equation*}
\begin{split}
&{\rm{Res}} (f_w,0) = \cfrac{2 (w-1)^2}{w^2}, \\
&{\rm{Res}} (f_w,z_-)= \lim_{z \to z_-} (z-z_-) f_{w} (z)= \cfrac{1}{z_- -z_+} \cfrac{[(w-1) + (1-w^{-1})z_{-}^2]^2}{z_-^2}.
\end{split}
\end{equation*}

By using the fact that $z_-^2 = 2z_- -w$, we obtain that
\begin{equation}
\begin{split}
{\rm Res} (f_w, 0) +{\rm {Res}} (f_w, z_-) &= \frac{2(w-1)^2}{w^2}\Big[ 1 + \frac{2}{z_- -z_+} \Big]\\
&= \frac{2(w-1)^2}{w^2}\Big[ 1 - \frac{1}{\sqrt{2|\sin(\pi y)|}} e^{- \frac{i\pi}{2} [y-\frac{1}{2} {\rm sgn} (y)]}\Big].
\end{split}
\end{equation}
Finally, we have
\begin{equation}
\begin{split}
G(y)&= \frac{1}{4} \frac{(e^{2i \pi y}-1)^2}{e^{4i \pi y}}\left[ 1 - \tfrac{1}{\sqrt{2|\sin(\pi y)|}} e^{- \frac{i\pi}{2} [y-\frac{1}{2} {\rm sgn} (y)]} \right]\\
&= \cfrac{1}{2} |\pi y|^{3/2} [ 1+ i \, {\rm{sgn}} (y)] + O(|y|^2).
\end{split}
\end{equation}
\end{proof}

\begin{lemma}
\label{lem:I}
The function $I$ defined by (\ref{eq:I}) satisfies, for any $y\in \RR$,
\begin{equation*}
| I(y)| \le C |\sin (\pi y)|^{3/2},
\end{equation*}
where $C$ is a positive constant independent of $y$.
\end{lemma}

\begin{proof}
We compute $I$ by using the residue theorem. For any $y \in [-\frac{1}{2}, \frac{1}{2}]$ we denote by $w:=w(y)$ the complex number $w=e^{2i\pi y}$. Then, we have
\begin{equation}
\label{eq:E8}
I(y)= -\frac{1}{4 i \pi} \cfrac{w-1}{w} \oint_{\mc C} f_{w} (z) dz,
\end{equation}
where the meromorphic function $f_{w}$ is defined by
\begin{equation}
f_w (z)= \frac{(z-1)(z^2 +w)}{z^2 (z-z_+)(z-z_-)}
\end{equation}
with $z_\pm$ defined by (\ref{eq:zpm}). We recall that $|z_-| <1$ and $|z_+|>1$ so that by the residue theorem we have
$$I(y) = -\cfrac{w-1}{2w} \left[ {\rm{Res}} (f_w,0) +{\rm {Res}} (f_w,z_-)\right].$$
A simple computation shows that
\begin{equation*}
{\rm{Res}} (f_w,0) =1-2/w, \quad {\rm {Res}} (f_w,z_-) =1/z_-.
\end{equation*}
It follows that
\begin{equation*}
I(y)= - \cfrac{w-1}{2w} \Big[ \frac{1}{z_-} +1 -\frac{2}{w}\Big].
\end{equation*}
Replacing $w$ and $z_-$ by their explicit values we get the result.
\end{proof}

\begin{lemma}
\label{lem:J}
The $1$-periodic function $J$ defined by (\ref{eq:J}) satisfies, for any $y \in \RR$,
 \begin{equation}
|J(y)| \le  C |\sin (\pi y)|^{-1/2},
\end{equation}
where $C$ is a positive constant independent of $y$.
\end{lemma}

\begin{proof}
We compute $J$ by using the residue theorem. For any $y \in [-\frac{1}{2}, \frac{1}{2}]$ we denote by $w:=w(y)$ the complex number $w=e^{2i\pi y}$. Then, we have
\begin{equation}
J(y)= -\frac{1}{4 i \pi} \oint_{\mc C} f_{w} (z) dz,
\end{equation}
where the meromorphic function $f_{w}$ is defined by
\begin{equation}
f_w (z)= \frac{(z^2 +w)}{z^2 (z-z_+)(z-z_-)}
\end{equation}
with $z_\pm$ defined by (\ref{eq:zpm}). By the residue theorem, we get
$$J(y)=-\frac{1}{2} \left[ {\rm{Res}} (f_w, 0) + {\rm Res} (f_w, z_-) \right].$$
A simple computation shows that
\begin{equation*}
{\rm{Res}} (f_w, 0)= - \frac{w}{2}, \quad {\rm{Res}} (f_w, z_-)=\frac{2}{z_- (z_- -z_+)}.
\end{equation*}
By using the explicit expressions for $w$, $z_\pm$, we get the result.
\end{proof}

\begin{lemma}
\label{lem:K}
The $1$-periodic function $K$ defined by (\ref{eq:K}) satisfies, for any $y \in \RR$,
 \begin{equation}
|K(y)| \le  C |\sin (\pi y)|^{1/2},
\end{equation}
where $C$ is a positive constant independent of $y$.
\end{lemma}

\begin{proof}
We compute $K$ by using the residue theorem. For any $y \in [-\frac{1}{2}, \frac{1}{2}]$ we denote by $w:=w(y)$ the complex number $w=e^{2i\pi y}$. Then, we have
\begin{equation}
K(y)= \frac{1}{4 i \pi} \oint_{\mc C} f_{w} (z) dz
\end{equation}
where the meromorphic function $f_{w}$ is defined by
\begin{equation}
f_w (z)= \frac{(z-1)(z^2 +w)}{z^2 (z-z_+)(z-z_-)}
\end{equation}
with $z_\pm$ defined by (\ref{eq:zpm}). Recalling (\ref{eq:E8}), we see that
\begin{equation*}
K(y)= -\cfrac{w}{w-1} I(y),
\end{equation*}
 and by Lemma \ref{lem:I} the result follows.
\end{proof}

\begin{lemma}
\label{lem:W}
The $1$-periodic function $W$ defined by  \eqref{eq:W} is such that
\begin{equation*}
W(y) = O (|y|^{-3/2})
\end{equation*}
on $[-\tfrac{1}{2}, \tfrac{1}{2}]$.
\end{lemma}

\begin{proof}
It is possible to compute $W$ by using the residue theorem and get the estimate. Since we need only an upper bound we bypass the computations and give a rough argument. On  $[-\tfrac{1}{2}, \tfrac{1}{2}]$, if $y$ is not close to $0$, say $|y| \ge \delta$ with $0<\delta<1/2$, then the integral is bounded above by a constant $C(\delta) < \infty$. If $|y| \le \delta$, then we split the integral into two integrals according to $|x| \le c \delta$ or $|x| \ge c\delta$ with $0<c<1/ (2\delta)$ a constant. We have
\begin{equation*}
 \int_{|x| \ge c \delta}  \frac{dx}{\Lambda (y-x,x)^2 + \Omega (y-x, x)^2 }  \le   \int_{|x| \ge c \delta}  \frac{dx}{16 \sin^4 (\pi x)  } \le C(\delta).
\end{equation*}
It remains then to show that if $|y| \le \delta$ then
\begin{equation*}
 \int_{|x| \le c \delta}  \frac{dx}{\Lambda(y-x,x)^2 + \Omega (y-x, x)^2 }  \le  C(\delta) |y|^{-3/2}.
 \end{equation*}
 Since $|x|, |y|$ are small, a Taylor expansion can be used to estimate the behavior of the previous integral. It is straightforward that it behaves like
 \begin{equation*}
 \int_{|x| \le c \delta}  \frac{dx}{x^4 +y^2} = O (|y|^{-3/2}).
 \end{equation*}
\end{proof}

\section{Estimates involving Hermite functions}
\label{A.7}
In this Appendix we prove (\ref{ec5.17}). For simplicity, assume $\ell =2m$. Let $M\ge 1$ and let $I_m$ be defined by
\begin{equation*}
I_m= \int_{\frac{M}{\sqrt 2}}^\infty x^{2m} e^{-\frac{x^2}{2}}dx.
\end{equation*}
By successive integration by parts, we have
\begin{equation*}
\begin{split}
I_m &= \big(\tfrac{M}{\sqrt{2}} \big)^{2m-1} e^{-M^2 /4} + (2m-1) I_{m-1} \\
&=\ldots= \cfrac{(2m)!}{m!}\, \left\{ e^{-M^2 /4} \sum_{k=0}^{m-1} \tfrac{1}{2^k}\tfrac{(m-k)!}{(2m- (2k+1))!} \big(\tfrac{M}{\sqrt{2}} \big)^{2m- (2k+1)}  + \tfrac{I_1}{2^m} \right\}\\
&=\cfrac{(2m)!}{2^m m!}\, \left\{ e^{-M^2 /4} \tfrac{\sqrt{2}}{M}  \sum_{k=1}^{m} \tfrac{k!}{(2k+1)!} M^{2k} + I_1\right\}\\
& \le C \cfrac{(2m)!}{2^m m!}\, \left\{ e^{-M^2 /4} m M^{2m-1} + e^{- M^2 /2} \right\}\\
& \le  C'  \cfrac{(2m)!}{2^m m!}\, m M^{2m-1} e^{-M^2 /4}.
\end{split}
\end{equation*}

We start now with the following representation of the Hermite polynomials:
\begin{equation}
H_\ell(x) = \ell! \sum_{j=0}^{\lfloor \ell/2\rfloor} \frac{(-1)^j x^{\ell - 2j}}{2^j j!(\ell-2j)!}.
\end{equation}
For $|x| \geq 1$, $|x|^{\ell-2j} \leq |x|^\ell$ and therefore we have that $\big|H_\ell(x) \big| \leq \ell! x^\ell$. It follows that
\begin{equation}
\begin{split}
\int_M^\infty \big| H_{2m}(x)\big| e^{-\frac{x^2}{4}} dx
		&\leq 2^m (2m)! \, \sqrt{2}\, I_m \leq C \frac{(2m)!^2}{m!} m M^{2m -1} e^{-M^2 /4}.
\end{split}
\end{equation}

By Stirling's formula, we conclude that
\begin{equation*}
\lim_{ m \to + \infty} \int_{|x| \ge m^{\frac{1+\delta}{2}}} | f_{2m} (x)|\,  dx =0,
\end{equation*}
uniformly in $m$.
Moreover, by Cauchy-Schwarz's inequality, since $\int \, f_{2m} (x)^2\,  dx =1$, we have
\begin{equation*}
\int_{|x| \le m^{\frac{1+\delta}{2}}} | f_{2m} (x)|\,  dx \le\sqrt{2} \,  m^{\frac{1+\delta}{4}}.
\end{equation*}


Since $\delta$ is arbitrary, (\ref{ec5.17}) is proved for $\ell $ even. For $\ell$ odd, the computations are similar.

\bibliographystyle{plain}


\end{document}